\documentclass{article}

 \usepackage[nonatbib,preprint]{neurips_2026}

\usepackage[utf8]{inputenc} 
\usepackage[T1]{fontenc}    
\usepackage{hyperref}       
\usepackage{url}            
\usepackage{booktabs}       
\usepackage{amsfonts}       
\usepackage{nicefrac}       
\usepackage{microtype}      
\usepackage{xcolor}         
\usepackage{amsmath}
\usepackage[capitalize,noabbrev,nameinlink]{cleveref}
\usepackage{tikz}
\usepackage{subcaption}
\usepackage{todonotes}
\usepackage[sort,square,numbers]{natbib}
\usepackage{algorithm}
\usepackage{algpseudocode}
\usepackage{amsthm}
\usepackage{thmtools}
\usepackage{mathtools}
\usepackage{multirow}

\usetikzlibrary{positioning, matrix, shapes.geometric, trees}  
\DeclareMathOperator*{\argmax}{arg\,max}

\title{Test-time Reinforcement Learning \\in Imperfect Information Games}

\author{%
  Ondrej Kubicek  \\
  Czech Technical University in Prague\\
  Carnegie Mellon University\\ 
  \texttt{kubicon3@fel.cvut.cz} \\
  \And
  Viliam Lisý \\
  Czech Technical University in Prague \\
  Artificial Intelligence Center \\
  \texttt{viliam.lisy@fel.cvut.cz} \\
  \AND
  Tuomas Sandholm \\
  Carnegie Mellon University \\
  Strategy Robot, Inc. \\
  Strategic Machine, Inc.\\
  Optimized Markets, Inc.\\
  \texttt{sandholm@cs.cmu.edu } \\ 
}

\begin{document}
\maketitle

\newcommand{\RealNumbers}[0]{\mathbb{R}}
\newcommand{\BinaryNumbers}[0]{\mathbb{B}}
\newcommand{\NaturalNumbers}[0]{\mathbb{N}}
\newcommand{\Expectation}[0]{\mathbb{E}}

\newcommand{\Simplex}{\Delta}

\newcommand{\FOSGame}[0]{\mathcal{G}}
\newcommand{\GeneralSubgameId}[0]{S}
\newcommand{\Subgame}[1]{\FOSGame^{#1}}
\newcommand{\GadgetGame}[1]{\FOSGame^{G,\PublicState}_{#1}}

\newcommand{\PublicIndex}[0]{0}
\newcommand{\InitialIndex}[0]{\text{INIT}}
\newcommand{\BRIndex}[0]{BR}

\newcommand{\Players}[0]{\mathcal{N}}
\newcommand{\Player}[0]{i}
\newcommand{\ChancePlayer}[0]{c}
\newcommand{\OtherPlayer}[0]{j}
\newcommand{\Opponent}[0]{o}
\newcommand{\LastPlayer}[0]{N}
\newcommand{\PlayerFunction}[0]{p}

\newcommand{\WorldStates}[0]{\mathcal{W}}
\newcommand{\WorldState}[0]{w}
\newcommand{\InitWorldState}[0]{\WorldState^{\InitialIndex}}

\newcommand{\Histories}[0]{\mathcal{H}}
\newcommand{\History}[0]{h}
\newcommand{\InitHistory}[0]{\History^{\InitialIndex}}
\newcommand{\HistoryExtend}{\sqsubseteq}

\newcommand{\Actions}[1]{\mathcal{A}_{#1}}
\newcommand{\Action}[1]{a_{#1}}
\newcommand{\ActionLogit}[1]{A_{#1}}

\newcommand{\Observations}[1]{\mathcal{O}_{#1}}
\newcommand{\Observation}[1]{o_{#1}}
\newcommand{\ObservationSet}[1]{\mathbb{O}_{#1}}
\newcommand{\PublicObservations}[0]{\Observations{\PublicIndex}}
\newcommand{\PublicObservation}[0]{\Observation{\PublicIndex}}
\newcommand{\PublicObservationSet}[0]{\ObservationSet{\PublicIndex}}

\newcommand{\Rewards}[1]{\mathcal{R}_{#1}}
\newcommand{\Reward}[1]{r_{#1}}
\newcommand{\WeightedReward}[2]{R_{#1}^{#2}}

\newcommand{\Utility}[1]{u_{#1}}

\newcommand{\Transitions}[0]{\mathcal{T}}

\newcommand{\Infosets}[1]{\mathcal{S}_{#1}}
\newcommand{\Infoset}[1]{s_{#1}}
\newcommand{\PublicStates}[0]{\Infosets{\PublicIndex}}
\newcommand{\PublicState}[0]{\Infoset{\PublicIndex}}
\newcommand{\Abstracted}[1]{\overline{#1}}

\newcommand{\AuxiliaryInfoset}[1]{\Infoset{#1}^{G}}
\newcommand{\AuxiliaryInfosets}[1]{\Infosets{#1}^{G}}

\newcommand{\Trajectory}[0]{\tau}
\newcommand{\TrajectoryLength}[0]{l}
\newcommand{\TrajectoryStep}[0]{t}

\newcommand{\Timestep}[0]{t}
\newcommand{\UnrollStep}[0]{k}

\newcommand{\Strategy}[1]{\pi_{#1}}
\newcommand{\Strategies}[1]{\Pi_{#1}}
\newcommand{\SamplingStrategy}[1]{\mu_{#1}}
\newcommand{\Nash}[1]{\Strategy{#1}^*}
\newcommand{\EpsilonNash}[1]{\Strategy{#1}^{\epsilon}}
\newcommand{\BRStrategy}[1]{\Strategy{#1}^{\BRIndex}}
\newcommand{\BRSet}[1]{\BRIndex_{#1}} 
\newcommand{\BlueprintStrategy}[1]{\overline{\Strategy{#1}}}
\newcommand{\BRBlueprintStrategy}[1]{\overline{\Strategy{#1}^{\BRIndex}}}
\newcommand{\BRValue}[0]{\text{BRV}}

\newcommand{\BeforeStrategy}[1]{\hat{\Strategy{#1}}}
\newcommand{\OptimalContinuationStrategy}[1]{\hat{\Strategy{}}_{#1}^*}

\newcommand{\StrategyUtility}[2]{\Utility{#1}^{#2}}
\newcommand{\CounterfactualValue}[2]{v_{#1}^{#2}}
\newcommand{\CounterfactualActionValue}[2]{q_{#1}^{#2}}
\newcommand{\Regret}[1]{r_{#1}}

\newcommand{\Exploitability}[0]{\mathcal{E}}

\newcommand{\Reach}[1]{P^{#1}}
\newcommand{\PlayerReach}[2]{\Reach{#2}_{#1}}

\newcommand{\StrategyPrior}[1]{\Strategy{#1}^P}

\newcommand{\Sequence}[1]{\sigma_{#1}}
\newcommand{\Sequences}[1]{\Sigma_{#1}}

\newcommand{\SQFStrategy}[1]{\boldsymbol{\Strategy{#1}}}
\newcommand{\SQFValue}[1]{\boldsymbol{v_{#1}}}
\newcommand{\SQFUtility}[0]{\boldsymbol{A}}
\newcommand{\SQFInitRealization}[1]{\boldsymbol{f_{#1}}}
\newcommand{\SQFRealizationConstraints}[1]{\boldsymbol{F_{#1}}}

\newcommand{\SQFPrior}[0]{\boldsymbol{p}}

\newcommand{\PriorEpsilon}[0]{\epsilon}

\newcommand{\PerturbationBasis}[0]{\boldsymbol{B}}

\newcommand{\Transposition}[0]{\top}

\newcommand{\GadgetParameters}[0]{\Psi}
\newcommand{\NeuralParameters}[0]{\theta}

\newcommand{\Critic}[1]{v_{#1}}

\newcommand{\KLDiv}[2]{D_{\text{KL}}({#1}, {#2})}
\newcommand{\BregmanDiv}[1]{B_{#1}}
\newcommand{\Regularizer}[0]{\psi}
\newcommand{\RegularizationWeight}[0]{\eta}
\newcommand{\Lagrangian}[0]{\mathcal{L}}

\newcommand{\MaxReward}[0]{\Rewards{}^{\text{max}}}
\newcommand{\MinReward}[0]{\Rewards{}^{\text{min}}}
\newcommand{\GameValue}[0]{V}

\begin{abstract}
Test-time reasoning has significantly improved performance in domains ranging from games to language models. However, test-time policy changes with formal guarantees on the performance of the resulting strategy remain a challenge in two-player zero-sum imperfect-information games. Existing solutions are limited to tabular methods or single gradient step updates. In this work, we investigate policy-gradient algorithms as a method for scalable test-time reasoning. We extend the concept of gadget game, tabular technique for test-time search, to the reinforcement learning setting. Unlike prior approaches, we represent the gadget game implicitly by modified sampling and neural policy rather then explicitly by constructing it, thereby removing constraints on subgame size. Furthermore, we formally prove that, unlike prior tabular algorithms, regularized policy-gradient algorithms limit possible strategy degradation caused by test-time reasoning, even without the gadget games. Our evaluation across small- and large-scale games confirms that additional test-time training often substantially improves performance relative to the blueprint strategy. 
\end{abstract}
    
\section{Introduction}  
Test-time reasoning enables algorithms to allocate limited computation to enhance performance within the context actually encountered during algorithm execution. This paradigm has been crucial to successes in game-playing \citep{silver2018alphazero,brown2018libartus,moravcik2017deepstack,brown2019pluribus,schmid2023studentofgames,zhang2026obscuro,sokota2025stratego,bakhtin2022diplomacy} and it has demonstrated significant gains in large language models  \citep{yang2024reasoning,kojima2022reasoning,snell2025reasoning}. 

In imperfect-information games, test-time reasoning is challenging. Since players cannot observe the underlying game state, they must reason across a distribution of possible states.  An optimal strategy at any given decision depends on the player's belief over these states, which is itself a function of the strategies employed by all players acting prior to that decision. Having any fixed belief regarding an opponent's past decisions has been shown to lead to highly exploitable strategies \citep{gilpin2006poker,burch2014cfrd,moravcik2016maxmargin}.

\textit{Gadget games} were developed to employ additional reasoning in tabular game solving settings. In this work we will focus on \textit{Bayesian gadget game}, which assumes both players followed some fixed strategy in the past and \textit{resolving gadget game} that does not require any fixed beliefs about opponents' strategies. By starting with a \textit{blueprint strategy} and refining it at each encountered decision point, resolving gadget provides a theoretical guarantee that the resulting strategy will be no more exploitable than the blueprint \citep{burch2014cfrd,moravcik2016maxmargin,brown2017reachmaxmargin}. We call methods with this guarantee ``safe'' test-time reasoning, and techniques without it ``unsafe''. Bayesian gadget game is an unsafe technique.

To date, gadget games and test-time reasoning have been largely restricted to tabular approaches. However, policy-gradient algorithms have emerged as a strong non-tabular alternative for game-playing, and were the first to achieve expert-level and then superhuman performance in the large, imperfect-information game Stratego \citep{perolat2023stratego,sokota2025stratego}. The superhuman Stratego bot Ataraxos used the ``update-equivalence framework'', a test-time reasoning technique that simulates a single policy-gradient step~\citep{sokota2024update,sokota2025stratego}. While this ensures that the strategy's exploitability strategy does not degrade by more than a certain bound, it also limits potential gains by restricting reasoning to a single update.

In this work, we propose methods for test-time reasoning using policy-gradient algorithms that avoid the explicit construction of a gadget game. However, resolving gadget game introduces new actions to the game. We use a small temporary neural network to represent the strategy in these new action nodes.
Unlike some prior approaches that used policy gradient during training, and then solved explicitly constructed gadget games using tabular algorithms \citep{kubicek2024lookahead,kubicek2025lookahead}, our method employs the same algorithm for both training and test-time reasoning, significantly reducing design complexity and improving the scalability.
To address the challenge of sampling the initial states of the subgame from the distribution required by gadget games, we train an additional transformer in train time.

Additionally, we demonstrate that safety guarantees similar to those of the update-equivalence framework can be maintained beyond a single gradient step, provided the regularization policy within the policy-gradient remains constant. This effectively bounds both the potential improvement or degradation of the policy.

We evaluate our proposed approaches both in small benchmark games, where exploitability can be precisely computed, and in large-scale games where tabular methods are computationally infeasible. 
In small games, we verify that simulating gadget games substantially improves weaker blueprints and only negligibly degrades strong blueprints due to sampling noise.
In large games, we demonstrate that using test time reasoning against an opponent without it leads to a win rate of over 80\% in Battleship.
\section{Background}
The two-player zero-sum factored-observation stochastic game (FOSG) is a tuple $\FOSGame = (\Players,\WorldStates, \InitWorldState, \PlayerFunction, \Actions{}, \Transitions, \Rewards{1},  \Observations{})$. Here, $\Players = (1, 2, \ChancePlayer)$ is the set of players, including a chance player $\ChancePlayer$, $\WorldStates$ represents the set of underlying world states, and $\InitWorldState \in \WorldStates$ is the initial state. The player function $\PlayerFunction: \WorldStates \to 2^{\Players}$ identifies which players are active in a given state. $\Actions{} = \Actions{1} \times \Actions{2} \times \Actions{c}$ denotes the set of joint actions, where $\Actions{\Player}(\WorldState)$ is the set of legal actions for player $\Player$ in state $\WorldState$. Transitions between world states are modeled by $\Transitions: \WorldStates \times \Actions{} \to \WorldStates$ and rewards of player 1 by $\Rewards{1} : \WorldStates \times \Actions{} \to \RealNumbers$. In zero-sum games, the rewards given to player 2 are $\Rewards{2} := -\Rewards{1}$. The observation function $\Observations{}: \WorldStates \times \Actions{} \times \WorldStates \to \ObservationSet{}$ provides observations for all players. The observation set $\ObservationSet{} = \ObservationSet{\PublicIndex} \times \ObservationSet{1} \times \ObservationSet{2}$ consists of public observations $\ObservationSet{\PublicIndex}$ and player-specific private observations $\ObservationSet{\Player}$. Similarly, the observation function can be factored as $\Observations{} = (\Observations{\PublicIndex}, \Observations{1}, \Observations{2})$ \citep{kovarik2022fosg}.

A history $\History = \WorldState^0 \Action{}^0 \dots \Action{}^{\TrajectoryLength-1} \WorldState^{\TrajectoryLength} \in (\WorldStates \Actions{})^* \WorldStates$ is a finite sequence starting in $\WorldState^0 = \InitWorldState$ and at each timestep $\Timestep \leq \TrajectoryLength-1$ the transitions are consistent with the game rules $\WorldState^{\Timestep + 1} = \Transitions(\WorldState^{\Timestep}, \Actions{}^{\Timestep})$. We will use $\Histories$ to denote the set of all possible histories in the game. For any 2 histories $\History, \History' \in \Histories$, if $\History$ is a prefix of $\History'$, we will use following notation $\History \HistoryExtend \History'$. Every history ends with a uniquely identifiable state, so we will sometimes use history instead of the world state, such as $\Rewards{1}(\History, \Action{}) := \Rewards{1}(\WorldState^{\TrajectoryLength}, \Action{})$ and $\InitHistory := \InitWorldState$ will be the initial history. In imperfect-information games, players cannot observe the world state directly. Instead, histories that are indistinguishable to the player $\Player$ are grouped into an information set (infoset) $ \Infoset{\Player} \in \Infosets{\Player}$. Moreover, the states that are indistinguishable from the perspective of the outside observer are in the same public state $\PublicState \in \PublicStates$. A different perspective is that $\Infoset{\Player}$ is a collection of states in which the player $\Player$ has the same private information, and $\PublicState$ is a collection of states in which the public information is the same. The sets of all information sets and public states are $\Infosets{\Player}$ and $\PublicStates$, respectively. We will overload the notation and use $\Histories(\Infoset{\Player})$ to denote all histories that share the information set, $\Infosets{\Player}(\PublicState)$ to denote all information sets that share the public state, $\Infoset{\Player}(\History)$ to denote infoset or public state that corresponds to history $\History$. Since the player cannot distinguish states within the same infoset, it must have the same legal actions in every such state. We will use $\Actions{\Player}(\Infoset{\Player}) := \Actions{\Player}(\WorldState)$, such that $\Infoset{\Player}$ contains $\WorldState$.

A behavioral strategy $\Strategy{\Player} : \Infosets{\Player} \to \Simplex \Actions{\Player}$ of player $\Player$ maps the information set to the probability distribution over legal actions. We will also use $\Strategy{\Player}(\Infoset{\Player}, \Action{\Player})$ to denote the probability of playing action $\Action{\Player}$ under strategy $\Strategy{\Player}$. The joint strategy profile is then $\Strategy{} = (\Strategy{1}, \Strategy{2}, \Strategy{\ChancePlayer})$, where $\Strategy{\ChancePlayer}$ is always fixed by the game rules. Consider two histories $\History, \History'$, such that $\History \HistoryExtend \History'$. The probability that the history $\History'$ is reached under $\Strategy{}$, from $\History$ is $\Reach{\Strategy{}}(\History' | \History) = \prod_{\History \HistoryExtend \History'' \Action{} \WorldState \HistoryExtend \History'} \prod_{\Player \in \Players} \Strategy{\Player}(\Infoset{\Player}(\History''), \Action{\Player})$. This probability is often called reach. This can be decomposed into individual player contributions $\PlayerReach{\Player}{\Strategy{}}(\History' | \History) = \prod_{\History \HistoryExtend \History'' \Action{} \WorldState \HistoryExtend \History'} \Strategy{\Player}(\Infoset{\Player}(\History''), \Action{\Player})$. We will use a standard notation where $-\Player$ denotes all players except player $\Player$. We also use $\Reach{\Strategy{}}(\History) = \Reach{\Strategy{}}(\History | \InitHistory)$.

The expected utility for player $\Player$ under strategy $\Strategy{}$ is $\Utility{\Player}^{\Strategy{}}(\History) = \sum_{\History \HistoryExtend \History' \Action{}} \Reach{\Strategy{}}(\History'|\History)\Rewards{\Player}(\History', \Action{}) \prod_{\OtherPlayer \in \Players} \Strategy{\OtherPlayer}(\Infoset{\OtherPlayer}(\History'), \Action{\OtherPlayer})$. The counterfactual utility of the information set  $\CounterfactualValue{\Player}{\Strategy{}}(\Infoset{\OtherPlayer}) = \frac{\sum_{\History \in \Histories(\Infoset{\OtherPlayer})} \PlayerReach{-\OtherPlayer}{\Strategy{}}(\History) \Utility{\Player}^{\Strategy{}}(\History)}{\sum_{\History \in \Histories(\Infoset{\OtherPlayer})} \PlayerReach{-\OtherPlayer}{\Strategy{}}(\History)}$ is the sum of expected utilities across all histories in an infoset, weighted by the probability that those histories are reached assuming player $\OtherPlayer$ played to reach them. The counterfactual action utility $\CounterfactualActionValue{\Player}{\Strategy{}}(\Infoset{\OtherPlayer}, \Action{\OtherPlayer})$ is defined analogously to the counterfactual utility, but for infoset-action pairs. A best response $\BRStrategy{\Player} \in \BRSet{\Player}(\Strategy{-\Player})$ maximizes utility against fixed opponent strategy $\BRStrategy{\Player} = \argmax_{\Strategy{\Player}} \Utility{\Player}^{(\Strategy{\Player}, \Strategy{-\Player})}(\InitHistory)$. A strategy profile where all players play a best response is a Nash equilibrium $\Nash{}$, which is a sought-after solution concept in two-player zero-sum games. Exploitability $\Exploitability(\Strategy{\Player}) = \Utility{-\Player}^{\Strategy{\Player}, \BRStrategy{-\Player}} (\InitHistory) - \Utility{-\Player}^{\Nash{}}(\InitHistory)$ measure a strategy's distance from the Nash equilibrium. It is always non-negative, and it is 0 iff the $\Strategy{\Player}$ is part of the Nash equilibrium.

\subsection{Subgame solving}
A subgame $\Subgame{\GeneralSubgameId}$ is a restriction of the game $\FOSGame$, which starts from a collection of states $\WorldStates^{\GeneralSubgameId}$ , which correspond to histories $\Histories^{\GeneralSubgameId} \subseteq \Histories$ from the original game $\FOSGame$. Based on the choice of $\Histories^{\GeneralSubgameId}$, there are 2 main categories of subgame solving. \textit{Common-knowledge subgame solving} considers all histories consistent with the public state $\Histories^{\GeneralSubgameId} = \Histories(\PublicState)$~\citep{gilpin2006poker,ganzfried2015endgame,burch2014cfrd, moravcik2016maxmargin, brown2017reachmaxmargin, kubicek2026refinements}. \textit{Knowledge-limited subgame solving} (KLSS) considers only a subset of public state histories~\citep{zhang2021subgame, zhang2026obscuro, liu2023opponent, sokota2024update}. Original KLSS ``freezes'' the strategy at the boundary information sets to equal the probabilities from the blueprint. We use a \textit{knowledge-limited unfrozen subgame solving (KLUSS)}, which allows both players to optimize their strategies across all states within the subgame~\citep{zhang2026obscuro}. Specifically, we focus on 1-KLUSS, which considers histories consistent with the information set of a player $\Histories^{\GeneralSubgameId} = \Histories(\Infoset{\Player})$. 

The subgame in imperfect information games is not a well-defined game, as it does not start in a single state. In order to make it a well defined game, it is necessary to connect the disjointed initial states with some new states and actions. We will call these new states and actions that were not part of the original game as \textit{gadget} and the resulting game \textit{gadget game}.

\textit{Bayesian gadget game} creates a new initial state, where chance plays to all states in $\WorldStates^{\GeneralSubgameId}$ proportionally to reach probabilities $\Reach{\BlueprintStrategy{}}(\History)$ of some blueprint strategy $\BlueprintStrategy{}$. Although often effective \citep{gilpin2006poker,gilpin2007abstraction,ganzfried2015endgame,brown2017reachmaxmargin,kubicek2026refinements}, this approach can lead to high exploitability if the opponent does not play the strategy assumed in the blueprint even if $\BlueprintStrategy{}$ is a Nash equilibrium \citep{burch2014cfrd}. \textit{Resolving gadget game} address this by providing the opponent with a hypothetical choice to modify their previous strategy, ensuring that the resulting solution is not more exploitable than the original blueprint \citep{burch2014cfrd,moravcik2016maxmargin,brown2017reachmaxmargin}. 

\subsection{Policy-gradient algorithms for games}
Policy-gradient algorithms have shown remarkable performance in very large games like Stratego \citep{perolat2023stratego,sokota2025stratego}. To ensure convergence in two-player zero-sum games, those methods employ KL regularization with respect to some \textit{magnet} (regularization) strategy $\Strategy{}^M$. In Regularized Nash Dynamics the regularized rewards are

\begin{equation}
\label{eq:regularized_reward}
    \Rewards{\Player}^R(\History, \Action{}) = \Rewards{\Player}(\History, \Action{}) - \RegularizationWeight \log(\frac{\Strategy{\Player}(\Infoset{\Player}(\History),\Action{\Player})}{\Strategy{\Player}^M(\Infoset{\Player}(\History),\Action{\Player})}) + \RegularizationWeight \log(\frac{\Strategy{-\Player}(\Infoset{-\Player}(\History),\Action{-\Player})}{\Strategy{-\Player}^M(\Infoset{-\Player}(\History),\Action{-\Player})})
\end{equation}

The corresponding optimization objective for behavioral strategies as defined by \citep{perolat2023stratego} is
\begin{equation}
\label{eq:reg_objective}
    \Strategy{\Player}^t(\Infoset{\Player}, \Action{\Player}) = \argmax_{\Strategy{\Player}} \sum_{\Action{\Player}} \Strategy{\Player}(\Infoset{\Player}, \Action{\Player}) q^{\Strategy{}^{t-1}}_{\Player}(\Infoset{\Player}, \Action{\Player}) - \eta \KLDiv{\Strategy{\Player}}{\Strategy{\Player}^M}
\end{equation}

A higher regularization strength ($\RegularizationWeight$) forces the strategy to remain closer to the magnet. This regularization alters the Nash equilibrium of the game based on the magnet strategy. It has been shown that by either changing the magnet strategy to slowly follow the computed strategy 
or by annealing the $\RegularizationWeight$, these algorithms eventually converge to a Nash equilibrium when used on normal-form or sequence-form strategies.  Like prior theoretical results \citep{perolat2021poincare, sokota2022mmd}, our results in \cref{thm:fixed_magnet_exploitability,thm:fixed_magnet_kl_divergence} apply to normal-form and sequence-form strategies. Still, the regularized policy-gradients exhibit surprisingly strong performance in large-scale games like Stratego even when applied to behavioral strategies~\citep{sokota2025stratego,perolat2023stratego}. Recently, it has been shown that these convergence guarantees hold even in behavioral strategies settings \citep{kalogiannis2026policy}. 

\section{Safe imperfect-information subgame solving}
\label{sec:example}

In imperfect-information games, subgame solving is difficult because player strategies are interdependent across the entire game. We illustrate this challenge using a modified version of Matching Pennies, in which Player 1 receives a double reward if both players select Tails. This game has a single unique Nash equilibrium if both players play $\Strategy{}(H) = \frac{2}{3}$.

\cref{fig:mp} shows the sequential form of this game, where Player 2 acts first without revealing the choice to Player 1. Consider Bayesian gadget game for Player 1's decision with 2 nodes as illustrated in \cref{fig:mp_subgame}. The optimal strategy in this subgame depends on the belief $p$ that the opponent played Heads.  If $p > \frac{2}{3}$, Player 1 should always play Heads; If $p < \frac{2}{3}$, then Tails is optimal. Lastly, if $p = \frac{2}{3}$, then regardless of the strategy the Player 1 chooses, it will always have the same expected utility $\frac{2}{3}$, which means that any strategy of Player 1 in the Bayesian gadget is optimal.

\begin{figure*}[!h]
    \centering    
    \begin{subfigure}[b]{0.32\textwidth}
        \centering
        \begin{tikzpicture}[
    level 1/.style={sibling distance=2.5cm, level distance=1cm},
    level 2/.style={sibling distance=1cm, level distance=1cm},
    p2node/.style={ 
      regular polygon,
      regular polygon sides=3,
      rotate=180, 
      draw,
      fill=blue!40,
      inner sep=1pt,
      minimum size=15pt
    },
    p1node/.style={
      regular polygon,
      regular polygon sides=3, 
      draw,
      fill=red!40,
      inner sep=1pt,
      minimum size=15pt
    },
    payoff/.style={
      font=\normalsize
    }
  ]

  \node [p2node] (root) {}
    child {
      node [p1node] (p1_left) {} 
      child { node [payoff] {1}  edge from parent node [left] {H} } 
      child { node [payoff] {0} edge from parent node [right] {T} }
      edge from parent node [above left] {H} 
    }
    child {
      node [p1node] (p1_right) {} 
      child { node [payoff] {0}  edge from parent node [left] {H} } 
      child { node [payoff] {2} edge from parent node [right] {T} }
      edge from parent node [above right] {T} 
    };

  \draw [red, dashed, thick] (p1_left) -- (p1_right)
        node [midway, above, yshift=2mm, fill=white, inner sep=1pt, font=\small] {$\Infoset{1}$};

\end{tikzpicture}
        \caption{Original game}
        \label{fig:mp}
    \end{subfigure}
    \hfill 
    \begin{subfigure}[b]{0.32\textwidth}
        \centering
        \begin{tikzpicture}[
    level 1/.style={sibling distance=2.5cm, level distance=1cm},
    level 2/.style={sibling distance=1cm, level distance=1cm},
    p2node/.style={ 
      regular polygon,
      regular polygon sides=3,
      rotate=180, 
      draw,
      fill=blue!40,
      inner sep=1pt,
      minimum size=15pt
    },
    p1node/.style={
      regular polygon,
      regular polygon sides=3, 
      draw,
      fill=red!40,
      inner sep=1pt,
      minimum size=15pt
    },
    chancenode/.style={circle, draw, fill=gray!30, minimum size=5mm, inner sep=1pt},
    payoff/.style={
      font=\normalsize
    }
  ]

  \node [chancenode] (root) {}
    child {
      node [p1node] (p1_left) {} 
      child { node [payoff] {1}  edge from parent node [left] {H} } 
      child { node [payoff] {0} edge from parent node [right] {T} }
      edge from parent node [above left] {$p$} 
    }
    child {
      node [p1node] (p1_right) {} 
      child { node [payoff] {0}  edge from parent node [left] {H} } 
      child { node [payoff] {2} edge from parent node [right] {T} }
      edge from parent node [above right] {$1-p$} 
    };

  \draw [red, dashed, thick] (p1_left) -- (p1_right)
        node [midway, above, yshift=2mm, fill=white, inner sep=1pt, font=\small] {$\Infoset{1}$};

\end{tikzpicture}
        \caption{Bayesian gadget game}
        \label{fig:mp_subgame}
    \end{subfigure}
    \hfill 
    \begin{subfigure}[b]{0.32\textwidth}
        \centering
        \begin{tikzpicture}[
    level 1/.style={sibling distance=2.5cm, level distance=1cm},
    level 2/.style={sibling distance=1cm, level distance=1cm},
    level 3/.style={sibling distance=1cm, level distance=1cm},
    p2node/.style={ 
      regular polygon,
      regular polygon sides=3,
      rotate=180, 
      draw,
      fill=blue!40,
      inner sep=1pt,
      minimum size=15pt
    },
    p1node/.style={
      regular polygon,
      regular polygon sides=3, 
      draw,
      fill=red!40,
      inner sep=1pt,
      minimum size=15pt
    },
    chancenode/.style={circle, draw, fill=gray!30, minimum size=5mm, inner sep=1pt},
    terminal/.style={inner sep=2pt},
    payoff/.style={
      font=\normalsize
    },
    edge_label/.style={midway, fill=white, inner sep=1pt}
]

\node[chancenode] (S) {} 
    child {
        node[p2node] (P2-L) {}
        child[grow=-150] { 
            node[payoff] (T1) {0.5}
            edge from parent
            node[edge_label, above left] {T}
        }
        child[grow=-90] { 
            node[p1node] (P1-L) {}
            child {
                node[payoff] (T3) {1}
                edge from parent
                node[edge_label, above left] {H}
            }
            child {
                node[payoff] (T4) {0}
                edge from parent
                node[edge_label, above right] {T}
            }
            edge from parent
            node[edge_label, right] {C}
        }
        edge from parent
        node[edge_label, above left] {1.0}
    }
    child {
        node[p2node] (P2-R) {}
        child[grow=-150] { 
            node[payoff] (T5) {1}
            edge from parent
            node[edge_label, above left] {T}
        }
        child[grow=-90] { 
            node[p1node] (P1-R) {}
            child {
                node[payoff] (T7) {0}
                edge from parent
                node[edge_label, above left] {H}
            }
            child {
                node[payoff] (T8) {2}
                edge from parent
                node[edge_label, above right] {T}
            }
            edge from parent
            node[edge_label, right] {C}
        }
        edge from parent
        node[edge_label, above right] {1.0}
    };
  \draw [red, dashed, thick] (P1-R) -- (P1-L)
        node [midway, above left, yshift=2mm, fill=white, inner sep=1pt, font=\small] {$\Infoset{1}$};

\end{tikzpicture}
        \caption{Resolving gadget game}
        \label{fig:mp_resolve}
    \end{subfigure}
    \caption{Sequential form of Matching Pennies, where Player 2 acts first and the corresponding subgame and resolving gadget game of Player 1.}
    \label{fig:example} 
\end{figure*}
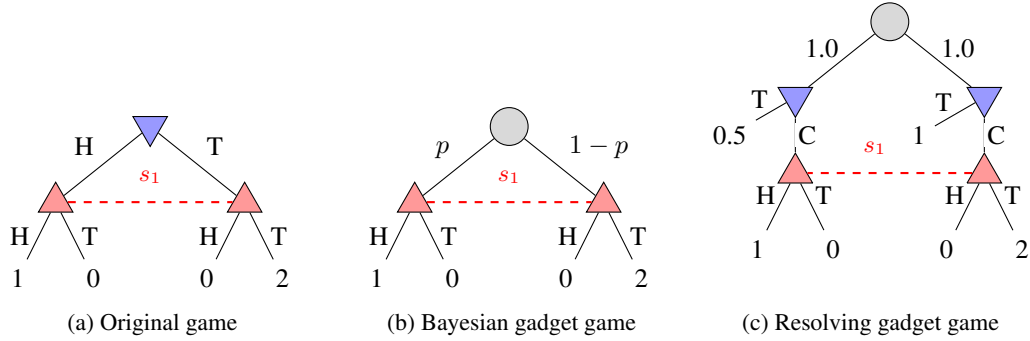 

The danger is that the opponent can choose a different strategy, such as playing H, which makes the different strategies of Player 1 in the real game lead to very different payoffs. Gadget games mitigate this risk by giving an opponent a choice to hypothetically select any strategy in the past. \textit{Resolving gadget} inserts an additional opponent's decisions before each root state of the subgame. These new decisions respect the opponent's information sets and give the opponent a choice to either continue into the subgame or terminate the game for a fixed reward. Crucially, the probability distribution across the root nodes is proportional only to the reach of resolving and the chance player \citep{burch2014cfrd}. 
 
If the termination reward is the opponent's counterfactual best response value to some blueprint strategy $\BlueprintStrategy{}$, then every equilibrium of this gadget game is guaranteed to be no more exploitable in the full game than the blueprint. Moreover, solutions to gadget games have been empirically shown to often improve the blueprint \citep{burch2014cfrd,brown2017reachmaxmargin,brown2018libartus,moravcik2017deepstack,kubicek2026refinements}.

\begin{restatable}[]{theorem}{GadgetValue}
\label{thm:gadget_value}
\citep{moravcik2017deepstack,burch2014cfrd,kubicek2026refinements} Consider a common-knowledge subgame $\Subgame{\PublicState}$ of game $\FOSGame$ starting in public state $\PublicState$. Let $\BlueprintStrategy{\Player}$ be a blueprint strategy and $\GadgetGame{\Player}$ be a player $\Player$'s resolving gadget game corresponding to the $\Subgame{\PublicState}$, where the terminate values correspond to the best response values to $\BlueprintStrategy{\Player}$. Let $\Nash{}$ be a Nash equilibrium of $\GadgetGame{\Player}$ and $\Strategy{\Player}'$ strategy that plays $\Nash{\Player}$ in a subgame $\Subgame{\PublicState}$ and blueprint $\BlueprintStrategy{}$ everywhere else. Then 
\[
    \Exploitability(\Strategy{\Player}') \leq \Exploitability(\BlueprintStrategy{\Player})
    \]
\end{restatable}

In \cref{fig:mp_resolve} we show the resolving gadget of the Matching Pennies subgame. The initial chance node assigns a probability of reaching the state equivalent to the reach of Player 1 and the chance. Since neither Player 1 nor the chance player has acted, the reach is 1. It is possible to normalize the chance node so that the probabilities sum to 1 without changing the optimal play in that game, but we kept them unnormalized so they reflect the reach. The reward of the terminating action corresponds to the blueprint $\BlueprintStrategy{1}(H) = 0.5$.

\subsection{Subgame solving by implicit gadget game}
Prior work bridges subgame solving and policy-gradients by training auxiliary components to explicitly construct subgames, which are then solved using tabular methods, such as CFR \citep{kubicek2024lookahead,kubicek2025lookahead}. However, this explicit construction of the complete subgame is often not computationally feasible.

Instead of explicit construction, we represent the resolving gadget game of player $\Player$ implicitly by introducing a temporary actor $\Strategy{\GadgetParameters}$ with weights $\GadgetParameters$. This actor predicts the strategy for the initial gadget decision (terminate or continue). The training in the gadget game proceeds by first sampling the initial state proportionally to the reach of player $\Player$ and the chance player. Then, it samples the trajectory until the end of the game using the current strategy $\Strategy{\NeuralParameters'}$. The q-value of the continue action is the estimated counterfactual value of the sampled initial state. The q-value of the terminate action is the critic's estimate using the blueprint strategy $\Strategy{\NeuralParameters}$. These two q-values are then used to compute the policy-gradient update. Lastly, the loss used to train $\NeuralParameters'$ is scaled by the probability of the continue action to ensure an unbiased update based on the opponent's gadget decision.

\cref{thm:gadget_value} does not depend directly on the solver used, as long as it converges to a Nash equilibrium, and our approach only mirrors the gadget game's behavior. As a result, our technique is not limited to regularized policy-gradient algorithms and will be compatible with 
future 
policy-gradient algorithms that may offer better theoretical guarantees for behavioral strategies.

While the ``terminate'' reward should, in theory, be the counterfactual best response value, we use the critic's estimate of the value when both players follow the blueprint. Some prior works observed that this approach performs better in practice, because it is less conservative \citep{moravcik2017deepstack}.

In nested subgame solving, gadget safety requires that counterfactual values and state samples remain consistent with the strategies used in previous subgames. Consequently, safety is only preserved if each subsequent resolve utilizes the weights derived from the preceding resolve.

\subsection{Subgame solving with a fixed magnet}

The most successful applications of policy-gradients in imperfect-information games typically employ KL regularization, which prevents the strategy from deviating too far from a regularization strategy $\Strategy{}^M$, sometimes referred to as a magnet \citep{perolat2021poincare,perolat2023stratego,sokota2022mmd,sokota2025stratego}. This regularization offers a distinct form of safety, which is not present in tabular solvers like CFR: the regularization strength $\RegularizationWeight$ explicitly bounds how far the strategy can shift.

\begin{restatable}[]{theorem}{FixedMagnetKLDivergence}
\label{thm:fixed_magnet_kl_divergence}
Given a game $\FOSGame$ and 1-strongly convex regularizer $\Regularizer$ with respect to $||\cdot||$. Let $\Strategy{}^M$ be a magnet strategy profile with exploitability $\epsilon$, $\RegularizationWeight > 0$ be a regularization weight. When solving a game $\FOSGame$ with regularized rewards, the Bregman divergence between the Nash strategy profile of the regularized game and the magnet is at most
\[
    \BregmanDiv{\Regularizer}(\Strategy{1}^\RegularizationWeight, \Strategy{1}^M) + \BregmanDiv{\Regularizer}(\Strategy{2}^\RegularizationWeight, \Strategy{2}^M) \leq \frac{\epsilon}{\RegularizationWeight}
\]
\end{restatable}
\begin{restatable}[]{theorem}{FixedMagnetExploitability}
\label{thm:fixed_magnet_exploitability}
Given a game $\FOSGame$ with rewards represented by a matrix $\Rewards{}$ and entropy as a regularizer $\Regularizer$. Let $\Strategy{}^M$ be a magnet strategy profile with exploitability $\epsilon$, $\RegularizationWeight > 0$ be a regularization weight. When solving a game $\FOSGame$ with regularized rewards, the exploitability of the Nash strategy profile of the regularized game in the unregularized game $\FOSGame$ is at most
\[
    \Exploitability(\Strategy{}^{\RegularizationWeight}) \leq \epsilon + ||\Rewards{}||_{\infty}  \sqrt{\frac{4 \epsilon}{\RegularizationWeight}} 
\]
\end{restatable}

Proofs are provided in \cref{app:proof}. \cref{thm:fixed_magnet_kl_divergence,thm:fixed_magnet_exploitability} describe the maximal possible strategy change for a fixed magnet $\Strategy{}^M$. Specifically, \cref{thm:fixed_magnet_exploitability} describes the change in exploitability when the algorithm is used in the entire game. When the fixed magnet is used in a Bayesian subgame, the exploitability is bound in that subgame, but when the same strategy is used in the rest of the game, the exploitability can increase beyond this bound. However, in practice, a strong blueprint $\BlueprintStrategy{}$ often implies low exploitability within each individual Bayesian subgame, and that bounds how much the strategy can change. So when the blueprint is near-optimal, the strategy changes in the subgame are outweighed by the KL penalty. Consequently, even when using the unsafe techniques, the strategy does not converge too far from the magnet.

The update-equivalence framework offers a similar perspective by simulating a single step of the policy gradient algorithm at each decision \citep{sokota2024update}. However, the change in exploitability in the update-equivalence framework is limited by the step size. In contrast, we show that this type of safety is maintained even upon convergence, provided the magnet is fixed. Still, both approaches may increase exploitability in the worst case.

Another surprising implication of \cref{thm:fixed_magnet_kl_divergence} is that it bounds the strategy change both in common-knowledge and knowledge-limited subgame solving. This is unlike the safety guaranteed by gadget games, which requires additional constraints to be extended to KLSS \citep{zhang2021subgame,liu2023opponent}.
\section{Generating states}
While the previous sections focus on strategy optimization, those methods require sampling states proportional to a specific probability distribution. As strategies are updated during subsequent resolves, this distribution shifts, requiring a belief model that dynamically changes along the strategy.

Prior approaches typically relied on explicit state enumeration using game rules \citep{gilpin2006poker,ganzfried2015endgame,brown2018depth,moravcik2017deepstack,brown2020rebel,schmid2023studentofgames,zhang2026obscuro,kubicek2024lookahead} or some form of abstraction that could be both trained or predefined \citep{brown2018libartus,kubicek2025lookahead,solinas2025neural}. More recently, belief models have been trained to sample states using a fixed strategy and potentially domain-specific knowledge \citep{sokota2024update,sokota2025stratego}. We propose a general approach: a generative model that learns to sample subgame states conditioned on available information.

Our belief model utilizes a decoder-only Transformer \citep{vaswani2017attention}. Rather than predicting a state representation directly, the model auto-regressively generates sequences of actions, including chance events, that lead to a state. Each action is represented as a discrete token. To ensure samples are consistent with the current context (information), we prepend either the infoset $\Infoset{\Player} $ (for KLSS), or the public state $\PublicState$ (for common-knowledge) as the initial token. The model then predicts the probability distribution over subsequent actions in the sequence.

The training target for the belief model depends on the subgame solving type. When using Bayesian gadget, the model must sample states proportional to $\Reach{\Strategy{\NeuralParameters}}(\WorldState | \Infoset{})$, the probability of reaching state $\WorldState$ conditioned on the $\Infoset{}$, which is either an information set or public state, under strategy $\Strategy{\NeuralParameters}$. We train this model by sampling trajectories using $\Strategy{\NeuralParameters}$ and applying standard cross-entropy loss to the generated action sequence.

For a resolving gadget of player $\Player$, the model should sample states proportionally to $\PlayerReach{\Player}{\Strategy{\NeuralParameters}}(\WorldState | \Infoset{}) \PlayerReach{\ChancePlayer}{\Strategy{\NeuralParameters}}(\WorldState | \Infoset{})$. Since standard trajectory sampling follows full reach $\Reach{\Strategy{\NeuralParameters}}(\WorldState | \Infoset{}) = \PlayerReach{1}{\Strategy{\NeuralParameters}}(\WorldState | \Infoset{}) \PlayerReach{2}{\Strategy{\NeuralParameters}}(\WorldState | \Infoset{}) \PlayerReach{\ChancePlayer}{\Strategy{\NeuralParameters}}(\WorldState | \Infoset{})$, which includes the opponent's $\OtherPlayer$ strategy, we apply the importance sampling correction $\frac{1}{\PlayerReach{\OtherPlayer}{\Strategy{\NeuralParameters}}(\WorldState | \Infoset{})}$ to remove bias from the samples of the opponent's specific strategy.

In a nested subgame-solving setting, the strategy is updated at every encountered decision point, while respecting past changes. This means that even the belief model needs to reflect those changes. During subgame solving, we train the belief model using the policy-gradient algorithm, ensuring it reflects changes in the strategy. Failing to do this joint training would result in a belief-strategy mismatch, where the strategy is optimized against an outdated distribution, potentially leading to significant bias in long games.
\section{Experiments}
We have evaluated various subgame solving techniques using Regularized Nash Dynamics (RNaD) with Neural Replicator Dynamics as the policy-gradient algorithm. The goal was to assess how these techniques influence exploitability in small-scale games and how they scale to large environments. Detailed information is in \cref{app:exp_details}

\subsection{One-shot games}
\label{sec:tiny_games}

To empirically verify  \cref{thm:fixed_magnet_exploitability,thm:gadget_value}, we solved subgames using either the Bayesian or the gadget approach. For Bayesian solving, we compared two configurations: either a moving magnet, where the magnet moves the same way as during full-game training, or a fixed magnet, aligning with constraints in \cref{thm:fixed_magnet_exploitability}. We tested these on two normal-form games, biased Matching Pennies (\cref{fig:mp}) and Rock-Paper-Scissors. We assumed that the players move sequentially in this game, as in the example in \cref{fig:mp}, and, as a result, each player has a single subgame. In both cases, the blueprint strategy was close to the Nash equilibrium $\Exploitability(\BlueprintStrategy{}) < 0.02$, which is the hardest setting for the subgame solving, because any small noise results in performance degradation. Consider a situation where the blueprint assigns a slightly higher probability to playing Paper than to other actions. Then the Bayesian solving would converge to a strategy that mixes between Scissors and Paper, which in turn can be exploited by playing Scissors.

Results in \cref{fig:tiny_games} confirm that a fixed magnet prevents the strategy from diverging too far from the equilibrium. In contrast, when moving the magnet, the strategy will gradually converge to the best response, which is exploitable in the full game. When using the gadget game, exploitability first increases before stabilizing. This happens because the gadget actor is initialized randomly and must first learn the strategy for the gadget decision. However, after this strategy is learned, the exploitability remains low, though it slightly increases compared to the blueprint. This increase is due to the approximation noise in the critic's predicted utilities for the gadget's ``terminate'' action.

\begin{figure*}[!h]
    \centering    
    \begin{subfigure}[b]{0.495\textwidth}
        \centering
    \includegraphics[width=0.99\linewidth]{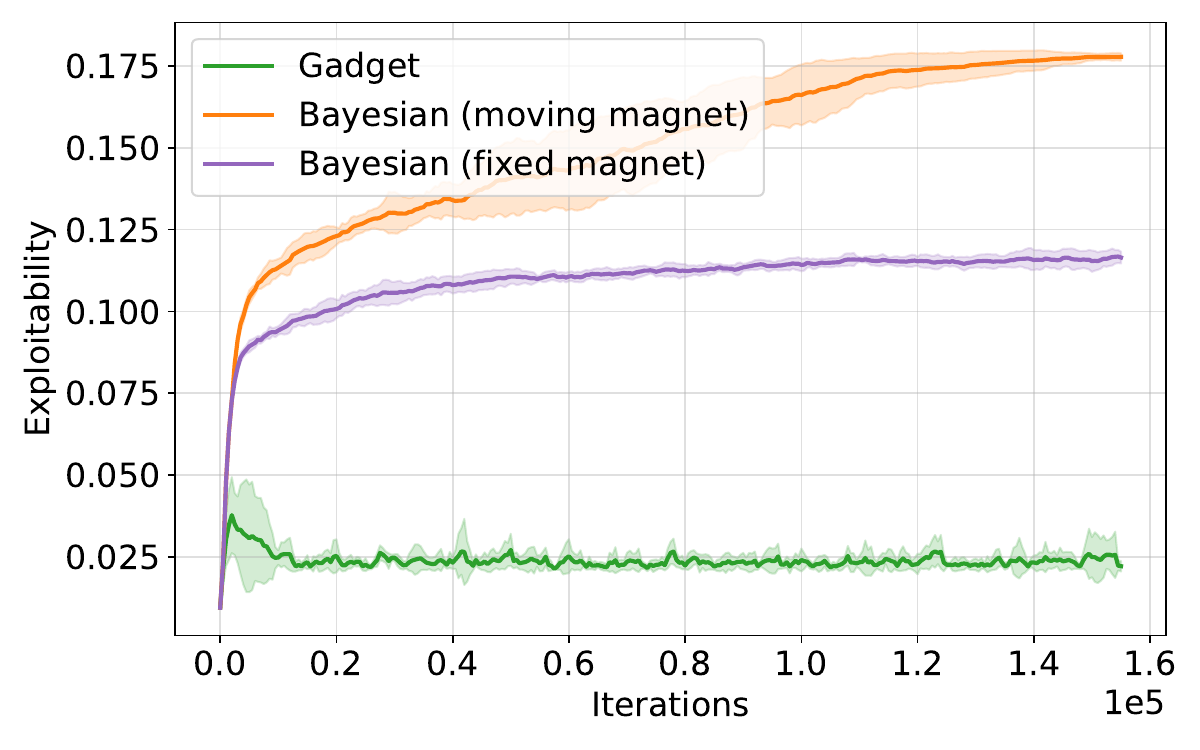}
        \caption{Biased matching Pennies}
        \label{fig:tiny_bmp}
    \end{subfigure}
    \hfill 
    \begin{subfigure}[b]{0.495\textwidth}
        \centering
    \includegraphics[width=0.99\linewidth]{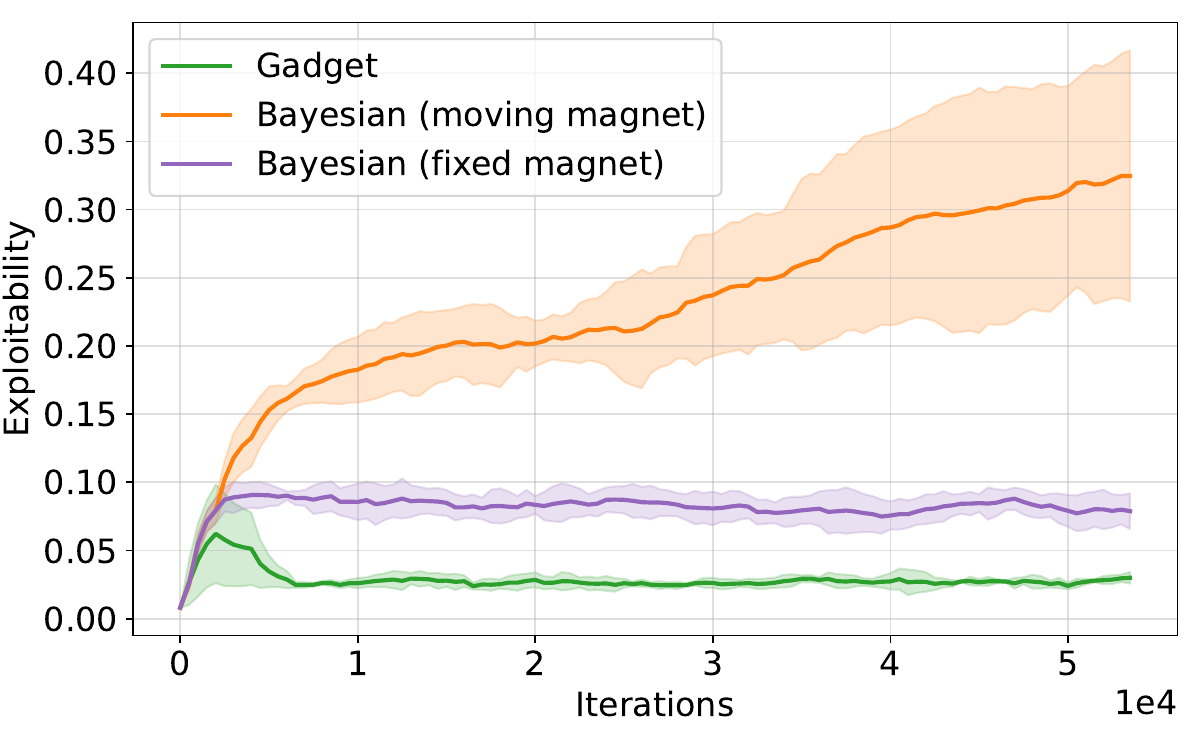}
        \caption{Rock-Paper-Scissors}
        \label{fig:tiny_rps}
    \end{subfigure} 
    \caption{Exploitability with 95\% confidence intervals in one-shot games, when solving a single subgame with different techniques using the same policy-gradient algorithm as the solver.}
    \label{fig:tiny_games} 
\end{figure*}  




\subsection{Small games}
\label{sec:small_games}

We next evaluated exploitability across full games by applying test-time reasoning at every decision point and then merging all subgame strategies into a single strategy for the whole game. We compared solutions of resolving gadget, common-knowledge Bayesian gadget, KLSS Bayesian gadget, and the Update Equivalence Framework (UEF), the method employed by the superhuman Stratego bot Ataraxos \citep{sokota2025stratego}. Since UEF is time-oblivious, we slightly modified it to use the provided time to collect traces, thereby obtaining a better estimate of Q-values for its update. When solving the Bayesian subgames, we kept the magnet fixed to be in line with \cref{thm:fixed_magnet_kl_divergence}.

Since later reasoning depends on the strategy from earlier subgames, we traversed the game using breadth-first search, and at each subgame we used the weights from the corresponding decision in the previous layer. We tested three standard benchmark games: Imperfect Information Goofspiel played with 5 cards, Leduc hold'em, and Battleship on 2x2 board with single ship of size 2, using three distinct blueprints for each. The first blueprint corresponds to the strategy near initialization of the weights, the second to the middle of the training, which will be the most common type of blueprint in large games, and the third, which is the least exploitable blueprint, to which the Regularized Nash Dynamics converged during training.

As shown in \cref{fig:small_games}, all techniques improve upon the blueprint when the blueprint is highly exploitable. Surprisingly, resolving gadget outperform other techniques with these blueprints, which is in contrast to the tabular settings, where Bayesian gadget is often superior \citep{brown2017reachmaxmargin,kubicek2026refinements}. This advantage stems partially from the fixed-magnet constraint in Bayesian gadget and the fact that RNaD naturally encourages non-zero action probabilities, which has recently been shown to converge to stronger equilibria in gadget games \citep{kubicek2026refinements}.

\begin{figure}
    \centering
    \includegraphics[width=0.99\linewidth]{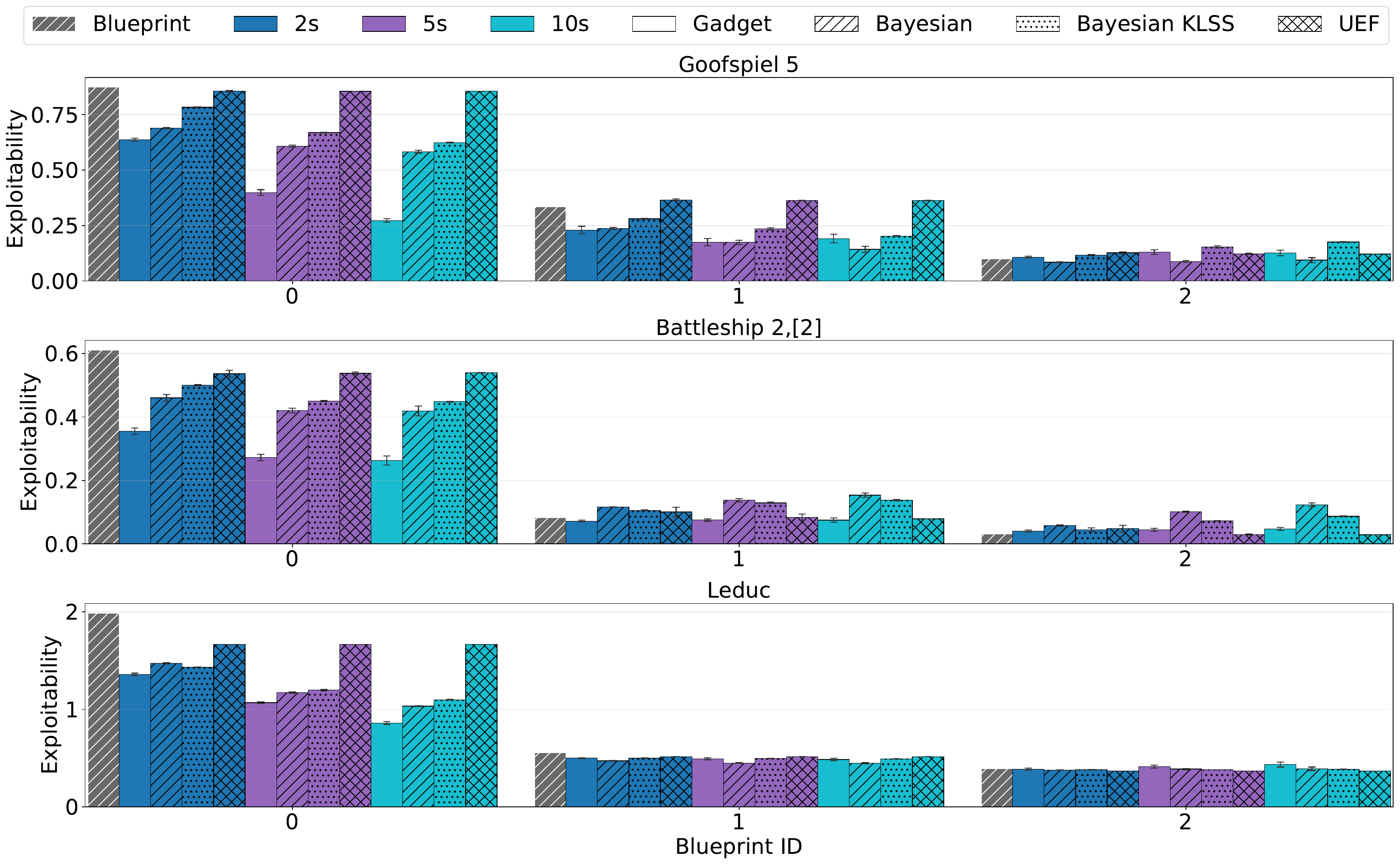}
    \caption{Exploitability with 95\% confidence interval after performing different types of test-time reasoning in every subgame with the same policy-gradient algorithm for a limited time.}
    \label{fig:small_games}
\end{figure}
Using the strongest blueprints in Goofspiel and Leduc, the solutions of resolving gadget seem to slightly increase exploitability even with extended train time. We hypothesize that there are two main reasons for this degradation. 1. Since the safety of gadget games is only guaranteed upon convergence, the allocated computation time was insufficient to reach equilibrium in earlier subgames, leading to compounding errors in the merged strategy. 2. The additional test-time training contains noise from several sources: the critic for the ``terminate'' action, samples from the belief model, and sampling of the trajectories.  

\subsection{Large games}
\label{sec:large_games}

Finally, we evaluated the same techniques in large-scale games where computing exact exploitability is intractable. We measured performance via head-to-head play against a baseline RNaD strategy. Both the subgame solver and the blueprint were allocated equal additional computation time at each decision point. The additional blueprint time was spent on further training the blueprint from the game's initial state. To ensure that the additional test-time reasoning does not exploit the knowledge of the opponent's strategy, we have always matched different training seeds.

The goal of this experiment is to verify that additional test-time reasoning improves upon the blueprint. We do not try to argue or show that this approach outperforms all prior techniques applicable to the used games.

A key implementation detail is that we reset the updated parameters to the blueprint after every playthrough. Early experiments suggested that without resetting, updates in one subgame would ``leak'' into others via shared neural network parameters, degrading global performance. Making these updates persistent without corrupting the global strategy remains a promising direction for future work. Even in tabular settings, the resolved strategy is used only temporarily and it does not update blueprint.

The results in \cref{fig:large_games} show that all methods outperform UEF. In Battleship, the UEF offered only marginal gains, likely because it is restricted to a single-step update. While Bayesian gadget is theoretically unsafe, it frequently outperformed the resolving gadget. This may be because resolving gadget weights the loss by the probability of ``continue'' action in the gadget decision, which effectively acts as a reduced learning rate. While common-knowledge approaches benefited from increased computation, KLSS appears to be hindered by it. Although head-to-head comparisons are not a perfect proxy for exploitability, these experiments demonstrate the potential for scalable test-time reasoning in large domains. 


\begin{figure*}
    \centering    
    \begin{subfigure}[b]{0.328\textwidth}
        \centering
    \includegraphics[width=0.99\linewidth]{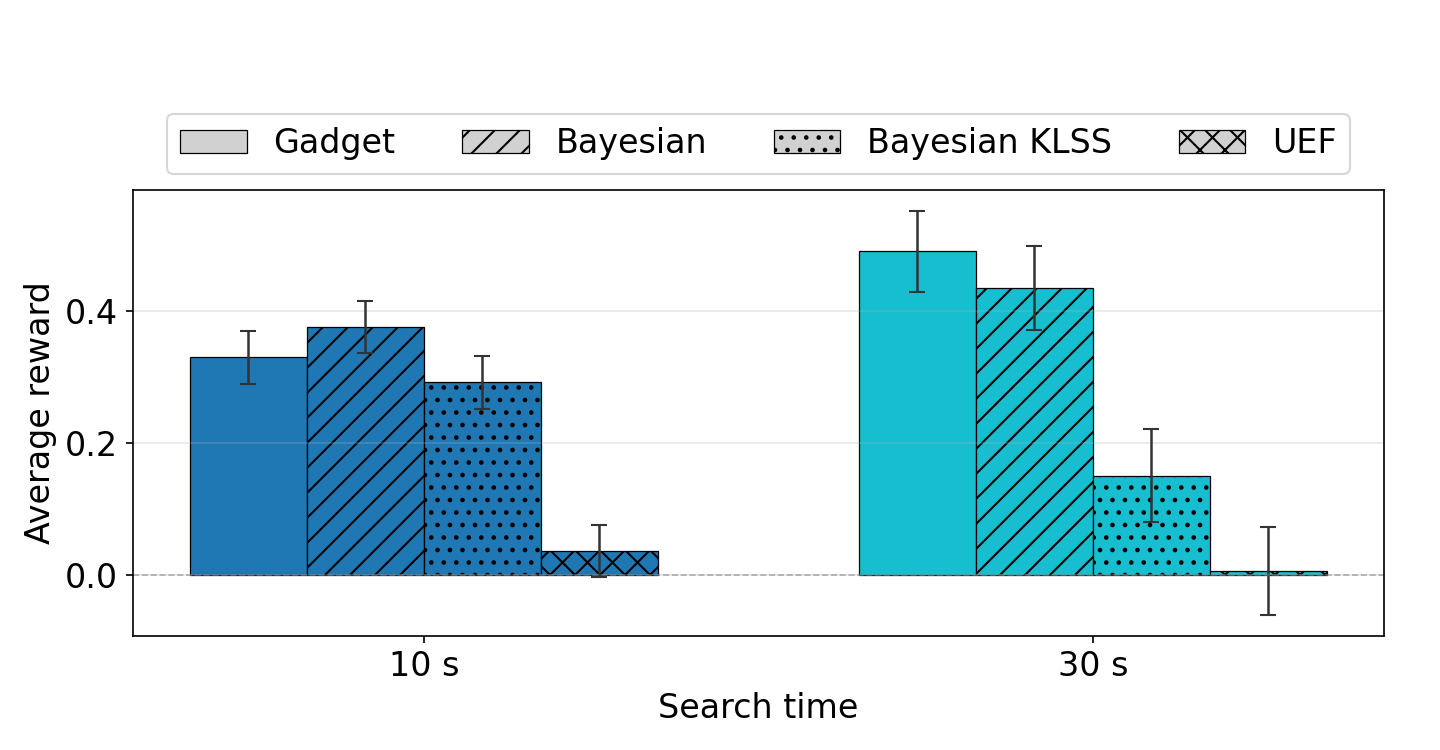}
        \caption{Battleship 7,[4,3,3,2]}
        \label{fig:bs7}
    \end{subfigure}
    \hfill 
    \begin{subfigure}[b]{0.328\textwidth}
        \centering
    \includegraphics[width=0.99\linewidth]{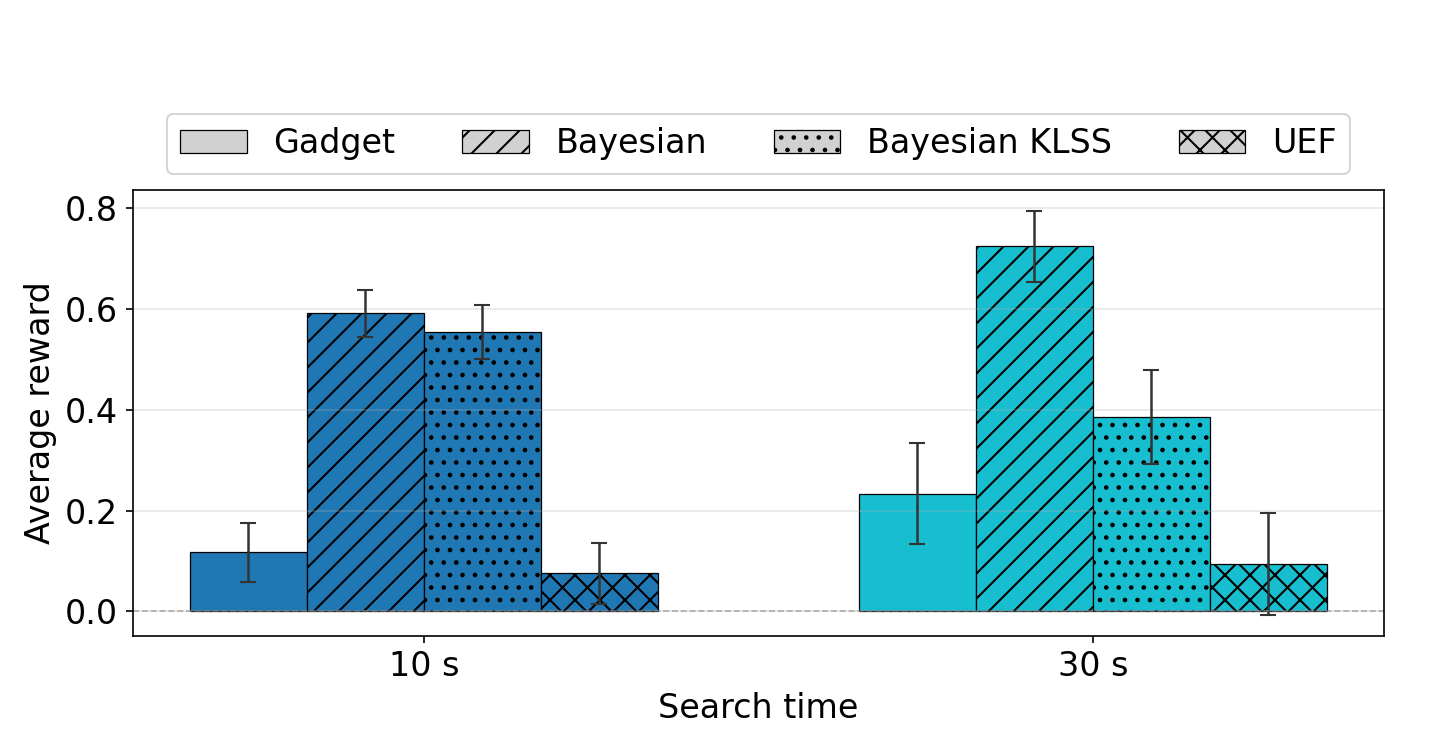}
        \caption{Battleship 10,[5,4,3,3,2]}
        \label{fig:bs10}
    \end{subfigure} 
    \hfill 
    \begin{subfigure}[b]{0.328\textwidth}
        \centering
    \includegraphics[width=0.99\linewidth]{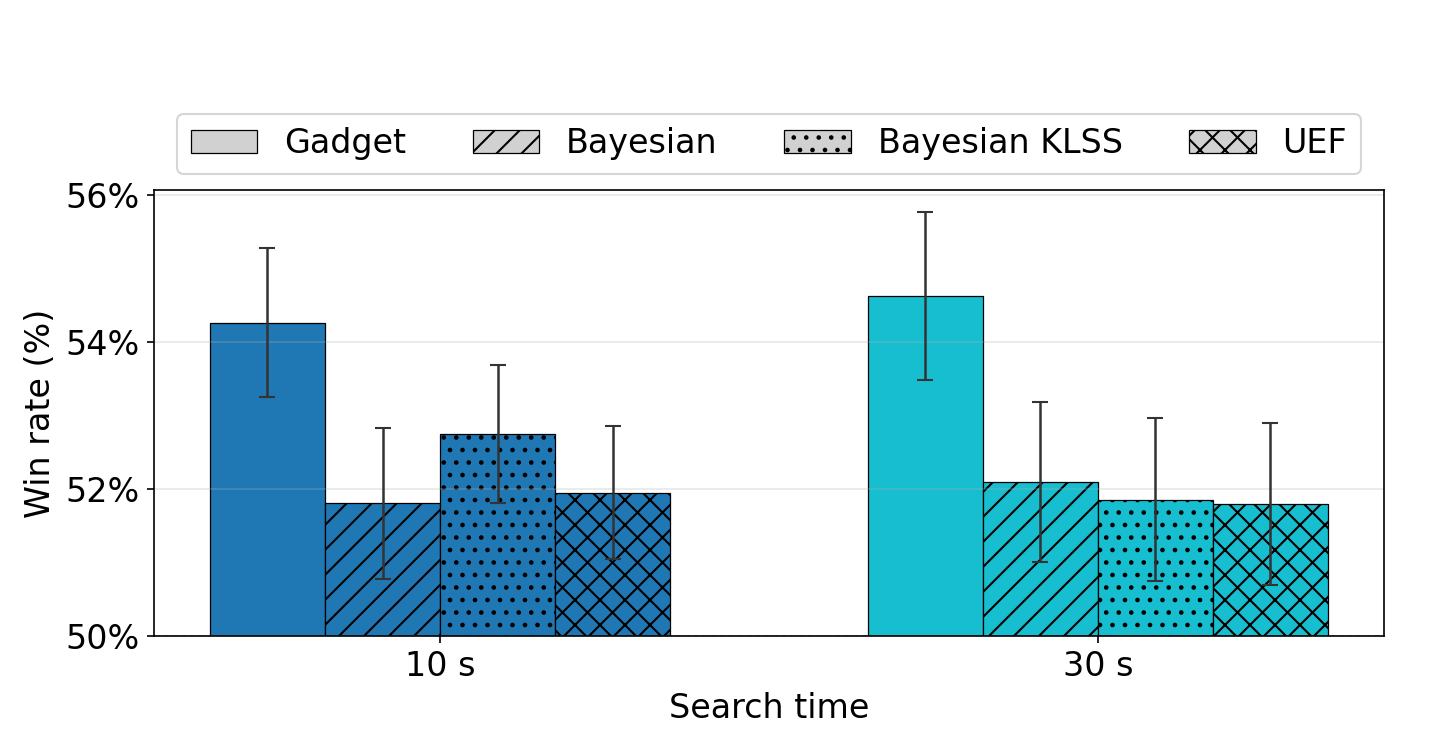}
        \caption{Goofspiel 10}
        \label{fig:hunl}
    \end{subfigure} 
    \caption{Head-to-head average reward with 95\% confidence intervals of different subgame solving techniques against the blueprint from different training seeds.}
    \label{fig:large_games} 
\end{figure*} 
\section{Conclusion and discussion}
In this work, we have presented a scalable framework for test-time reasoning in two-player zero-sum imperfect information games using policy-gradient algorithms. By extending the concept of gadget games to the reinforcement learning setting and introducing a generative belief model, we have removed the reliance on explicit subgame construction and tabular solvers. This allows test-time reasoning to scale to games, where subgame states may not be enumerable. Furthermore, because our approach only slightly modifies the core training algorithm, it significantly reduces design complexity compared to hybrid tabular-RL methods.

We demonstrated that resolving gadget game can be implicitly represented using a temporary actor that models the gadget's decision. We also established that KL regularization in modern policy-gradient algorithms serves as an implicit safety anchor, bounding how much the strategy changes even when using nominally unsafe Bayesian techniques.

We empirically evaluated these approaches in both small and large games, verifying that even in this setting, the additional test-time reasoning improves the performance over the blueprint. Notably, these techniques often yielded greater performance gains than the Update-equivalence framework (UEF) utilized in the superhuman Stratego bot Ataraxos \citep{sokota2024update,sokota2025stratego}.

We observed that the shared network parameters necessitate resetting weights after each playthrough. This highlights a potential for future research that would enable persistent local updates, allowing agents to retain test-time insights across playthroughs without degrading global strategy.

Another promising direction is to combine our approach with depth-limited solving to allow faster gradient updates in long games. Prior work has integrated an approximate imperfect-information value function with policy-gradient algorithms \citep{kubicek2024lookahead,kubicek2025lookahead}. We intentionally omitted such value functions to isolate the effects of test-time reasoning from the potential noise of value approximation. 
\newpage
\section*{Acknowledgments}
This work is supported by National Science Foundation grant RI-2312342, the Vannevar Bush Faculty Fellowship ONR N00014-23-1-2876, the Czech Science Foundation GA25-18353S, the Grant Agency of the Czech Technical University in Prague (SGS23/184/OHK3/3T/13). The access to the computational infrastructure of the OP VVV funded project CZ.02.1.01/0.0/0.0/16\_019/0000765 ``Research Center for Informatics'' is also gratefully acknowledged.  Any opinions, findings, and conclusions or recommendations expressed in this material are those of the author(s) and do not necessarily reflect the views of the funding agencies.
\bibliography{references}

\appendix
\section{Proofs}
\label[appendix]{app:proof}

\FixedMagnetKLDivergence*
We will follow the Magnetic Mirror Descent approach, which defines the regularization with respect to the Bregman divergence $\BregmanDiv{}$, which is more general than the KL divergence. We will denote $\Rewards{}$ as the reward matrix corresponding to \cref{eq:regularized_reward}. The strategies used throughout this proof refer to either the normal-form strategies or the sequence-form strategies, not the behavioral strategies. When using the fixed magnet $\Strategy{}^M = (\Strategy{1}^M, \Strategy{2}^M)$, the regularized objective can then be written of as
\begin{equation}
\label{eq:regularized_game_value}
 \GameValue(\Strategy{1}, \Strategy{2}) = \Strategy{1}^T \Rewards{} \Strategy{2} - \RegularizationWeight \BregmanDiv{\Regularizer}(\Strategy{1}, \Strategy{1}^M) +  \RegularizationWeight \BregmanDiv{\Regularizer}(\Strategy{2}, \Strategy{2}^M)
\end{equation}
For the remainder of this section, we will not denote the transposition of the strategy vector of player 1, when it is multiplied with the reward matrix, which will simplify the notation, but we will still treat it as transposed vector.
and the solution of such regularized game is
\begin{equation}
\label{eq:regularized_objective}
    \max_{\Strategy{1}} \min_{\Strategy{2}} \GameValue(\Strategy{1}, \Strategy{2})
\end{equation}

The fixed point solution $\Strategy{1}^\RegularizationWeight, \Strategy{2}^\RegularizationWeight$ of this minmax problem satisfies these conditions for each $\Strategy{1}$ and $\Strategy{2}$
\begin{align}
    \langle \nabla_{\Strategy{1}} \GameValue(\Strategy{1}^\RegularizationWeight, \Strategy{2}^\RegularizationWeight), \Strategy{1} - \Strategy{1}^\RegularizationWeight \rangle \leq 0 \label{eq:vi_player1}\\
    -\langle \nabla_{\Strategy{2}} \GameValue(\Strategy{1}^\RegularizationWeight, \Strategy{2}^\RegularizationWeight), \Strategy{2} - \Strategy{2}^\RegularizationWeight \rangle  \leq 0 \label{eq:vi_player2}
\end{align}

We can compute the gradient in \cref{eq:vi_player1}

\begin{align}
    \nabla_{\Strategy{1}} \GameValue(\Strategy{1}, \Strategy{2}) = \Rewards{}\Strategy{2} - \RegularizationWeight & \nabla_{\Strategy{1}}\BregmanDiv{\Regularizer}(\Strategy{1}, \Strategy{1}^M) = \Rewards{}\Strategy{2} - \RegularizationWeight \nabla \Regularizer(\Strategy{1}) + \RegularizationWeight \nabla \Regularizer(\Strategy{1}^M) \\
    \langle \nabla_{\Strategy{1}} \GameValue(\Strategy{1}^\RegularizationWeight, \Strategy{2}^\RegularizationWeight), \Strategy{1} - \Strategy{1}^\RegularizationWeight \rangle &= \langle \Rewards{}\Strategy{2}^\RegularizationWeight - \RegularizationWeight \nabla \Regularizer(\Strategy{1}^\RegularizationWeight) + \RegularizationWeight \nabla \Regularizer(\Strategy{1}^M), \Strategy{1} - \Strategy{1}^\RegularizationWeight \rangle \nonumber \\
    &= \langle \Rewards{}\Strategy{2}^\RegularizationWeight, \Strategy{1} - \Strategy{1}^\RegularizationWeight \rangle - \RegularizationWeight \langle  \nabla \Regularizer(\Strategy{1}^\RegularizationWeight) - \nabla  \Regularizer(\Strategy{1}^M), \Strategy{1} - \Strategy{1}^\RegularizationWeight \rangle
\end{align}
We will use the Bregman three-point property $\BregmanDiv{\Regularizer}(x, z) = \BregmanDiv{\Regularizer}(x, y) + \BregmanDiv{\Regularizer}(y, z) + \langle  \nabla \Regularizer(y) - \nabla  \Regularizer(z), x- y \rangle$
\begin{align}
   & \BregmanDiv{\Regularizer}(\Strategy{1}, \Strategy{1}^M) = \BregmanDiv{\Regularizer}(\Strategy{1}, \Strategy{1}^\RegularizationWeight) + \BregmanDiv{\Regularizer}(\Strategy{1}^\RegularizationWeight, \Strategy{1}^M) + \langle  \nabla \Regularizer(\Strategy{1}^\RegularizationWeight) - \nabla  \Regularizer(\Strategy{1}^M), \Strategy{1} - \Strategy{1}^\RegularizationWeight \rangle \nonumber\\
    &\langle  \nabla \Regularizer(\Strategy{1}^\RegularizationWeight) - \nabla  \Regularizer(\Strategy{1}^M), \Strategy{1} - \Strategy{1}^\RegularizationWeight \rangle = \BregmanDiv{\Regularizer}(\Strategy{1}, \Strategy{1}^M) -  \BregmanDiv{\Regularizer}(\Strategy{1}, \Strategy{1}^\RegularizationWeight) - \BregmanDiv{\Regularizer}(\Strategy{1}^\RegularizationWeight, \Strategy{1}^M)\\
    &\langle \nabla_{\Strategy{1}} \GameValue(\Strategy{1}^\RegularizationWeight, \Strategy{2}^\RegularizationWeight), \Strategy{1} - \Strategy{1}^\RegularizationWeight \rangle = \nonumber\\
    &= \langle \Rewards{}\Strategy{2}^\RegularizationWeight, \Strategy{1} - \Strategy{1}^\RegularizationWeight \rangle - \RegularizationWeight(\BregmanDiv{\Regularizer}(\Strategy{1}, \Strategy{1}^M) -  \BregmanDiv{\Regularizer}(\Strategy{1}, \Strategy{1}^\RegularizationWeight) - \BregmanDiv{\Regularizer}(\Strategy{1}^\RegularizationWeight, \Strategy{1}^M)) \leq 0\\
    &\Rightarrow \Strategy{1} \Rewards{}\Strategy{2}^\RegularizationWeight - \Strategy{1}^{\RegularizationWeight}\Rewards{}\Strategy{2}^\RegularizationWeight \leq  \RegularizationWeight(\BregmanDiv{\Regularizer}(\Strategy{1}, \Strategy{1}^M) - \BregmanDiv{\Regularizer}(\Strategy{1}, \Strategy{1}^\RegularizationWeight) -\BregmanDiv{\Regularizer}(\Strategy{1}^\RegularizationWeight, \Strategy{1}^M))
\end{align}
Similarly, we will derive for the second constraint
\begin{align}
    &-\langle \nabla_{\Strategy{2}} \GameValue(\Strategy{1}^\RegularizationWeight, \Strategy{2}^\RegularizationWeight), \Strategy{2} - \Strategy{2}^\RegularizationWeight \rangle \nonumber  \\
    &=-\langle \Strategy{1}^{\RegularizationWeight} \Rewards{}, \Strategy{2} - \Strategy{2}^\RegularizationWeight \rangle - \RegularizationWeight \langle \nabla \Regularizer(\Strategy{2}^\RegularizationWeight) - \nabla\Regularizer(\Strategy{2}^M) , \Strategy{2} - \Strategy{2}^\RegularizationWeight \rangle \nonumber \\
    & = -\langle \Strategy{1}^{\RegularizationWeight} \Rewards{}, \Strategy{2} - \Strategy{2}^\RegularizationWeight \rangle - \RegularizationWeight(\BregmanDiv{\Regularizer}(\Strategy{2}, \Strategy{2}^M) -  \BregmanDiv{\Regularizer}(\Strategy{2}, \Strategy{2}^\RegularizationWeight) - \BregmanDiv{\Regularizer}(\Strategy{2}^\RegularizationWeight, \Strategy{2}^M))  \leq 0\\
    &\Rightarrow  -\Strategy{1}^{\RegularizationWeight} \Rewards{} \Strategy{2} + \Strategy{1}^{\RegularizationWeight} \Rewards{} \Strategy{2}^\RegularizationWeight \leq \RegularizationWeight(\BregmanDiv{\Regularizer}(\Strategy{2}, \Strategy{2}^M) -  \BregmanDiv{\Regularizer}(\Strategy{2}, \Strategy{2}^\RegularizationWeight) - \BregmanDiv{\Regularizer}(\Strategy{2}^\RegularizationWeight, \Strategy{2}^M)) 
\end{align}

We can now sum those 2 inequalities
\begin{align}
\label{eq:mmd_vi_finished}
     \Strategy{1} \Rewards{} \Strategy{2}^\RegularizationWeight -\Strategy{1}^{\RegularizationWeight}\Rewards{}\Strategy{2}  &\leq \RegularizationWeight \big(\BregmanDiv{\Regularizer}(\Strategy{1}, \Strategy{1}^M)  +\BregmanDiv{\Regularizer}(\Strategy{2}, \Strategy{2}^M) + \nonumber \\
         &- \BregmanDiv{\Regularizer}(\Strategy{1}, \Strategy{1}^\RegularizationWeight) - \BregmanDiv{\Regularizer}(\Strategy{2}, \Strategy{2}^\RegularizationWeight) - \BregmanDiv{\Regularizer}(\Strategy{1}^\RegularizationWeight, \Strategy{1}^M) - \BregmanDiv{\Regularizer}(\Strategy{2}^\RegularizationWeight, \Strategy{2}^M) \big)
\end{align}

The \cref{eq:mmd_vi_finished} holds for every $\Strategy{} = (\Strategy{1}, \Strategy{2})$, so it has to holds even for $\Strategy{} = \Strategy{}^M$ and $\BregmanDiv{x, x} = 0$

\begin{equation} 
     \Strategy{1}^{M} \Rewards{} \Strategy{2}^\RegularizationWeight -\Strategy{1}^{\RegularizationWeight}\Rewards{}\Strategy{2}^M  \leq \RegularizationWeight \big(- \BregmanDiv{\Regularizer}(\Strategy{1}^M, \Strategy{1}^\RegularizationWeight) - \BregmanDiv{\Regularizer}(\Strategy{2}^M, \Strategy{2}^\RegularizationWeight) - \BregmanDiv{\Regularizer}(\Strategy{1}^\RegularizationWeight, \Strategy{1}^M) - \BregmanDiv{\Regularizer}(\Strategy{2}^\RegularizationWeight, \Strategy{2}^M)\big)
\end{equation} 

\begin{align} 
      \BregmanDiv{\Regularizer}(\Strategy{1}^M, \Strategy{1}^\RegularizationWeight) + \BregmanDiv{\Regularizer}(\Strategy{2}^M, \Strategy{2}^\RegularizationWeight) &+ \BregmanDiv{\Regularizer}(\Strategy{1}^\RegularizationWeight, \Strategy{1}^M) + \BregmanDiv{\Regularizer}(\Strategy{2}^\RegularizationWeight, \Strategy{2}^M) \nonumber\\ &\leq \frac{1}{\RegularizationWeight}(\Strategy{1}^{\RegularizationWeight}\Rewards{}\Strategy{2}^M - \Strategy{1}^{M} \Rewards{} \Strategy{2}^\RegularizationWeight) \nonumber \\
      &\leq \frac{1}{\RegularizationWeight} (\max_{\Strategy{1}} \Strategy{1}\Rewards{}\Strategy{2}^M - \min_{\Strategy{2}} \Strategy{1}^{M} \Rewards{} \Strategy{2}) = \frac{\Exploitability(\Strategy{}^M)}{\RegularizationWeight}
\end{align}

We know that the Bregman divergence is always non-negative, which allows us to get a bound on the resulting KL diergence for the strategy profile

\begin{equation}
\label{eq:mmd_expl_bound}
    \BregmanDiv{\Regularizer}(\Strategy{1}^\RegularizationWeight, \Strategy{1}^M) + \BregmanDiv{\Regularizer}(\Strategy{2}^\RegularizationWeight, \Strategy{2}^M) \leq \frac{\Exploitability(\Strategy{}^M)}{\RegularizationWeight}
\end{equation}
\FixedMagnetExploitability*

\begin{proof}  
Exploitability of the $\Strategy{}^{\RegularizationWeight}$ is 
\begin{align}
    &\Exploitability(\Strategy{}^{\RegularizationWeight}) = \Strategy{1}^{\BRIndex} \Rewards{} \Strategy{2}^\RegularizationWeight -\Strategy{1}^{\RegularizationWeight}\Rewards{}\BRStrategy{2} \nonumber \\
    &= \Strategy{1}^{\BRIndex} \Rewards{} \Strategy{2}^\RegularizationWeight -\Strategy{1}^{\RegularizationWeight}\Rewards{}\BRStrategy{2} + \Strategy{1}^{\BRIndex} \Rewards{} \Strategy{2}^M - \Strategy{1}^{\BRIndex} \Rewards{} \Strategy{2}^M  + \Strategy{1}^{M} \Rewards{} \BRStrategy{2}- \Strategy{1}^{M} \Rewards{} \BRStrategy{2} \nonumber \\
    &= \Strategy{1}^{\BRIndex} \Rewards{} \Strategy{2}^\RegularizationWeight - \Strategy{1}^{\BRIndex} \Rewards{} \Strategy{2}^M + \Strategy{1}^{\BRIndex} \Rewards{} \Strategy{2}^M - \Strategy{1}^{M} \Rewards{} \BRStrategy{2} +  \Strategy{1}^{M} \Rewards{} \BRStrategy{2} - \Strategy{1}^{\RegularizationWeight}\Rewards{}\BRStrategy{2} \nonumber\\
    &= \Strategy{1}^{\BRIndex} \Rewards{}(\Strategy{2}^\RegularizationWeight - \Strategy{2}^M) + (\Strategy{1}^{\BRIndex} \Rewards{} \Strategy{2}^M - \Strategy{1}^{M} \Rewards{} \BRStrategy{2}) + (\Strategy{1}^{M} -\Strategy{1}^{\RegularizationWeight})\Rewards{}\BRStrategy{2} 
\end{align}
We can now use following bound for middle terms
\begin{equation}
    \Strategy{1}^{\BRIndex} \Rewards{} \Strategy{2}^M - \Strategy{1}^{M} \Rewards{} \BRStrategy{2} \leq \max_{\Strategy{1}} \Strategy{1} \Rewards{} \Strategy{2}^M - \min_{\Strategy{2}} \Strategy{1}^{M} \Rewards{} \Strategy{2} = \Exploitability(\Strategy{}^M)
 \end{equation}

We can bound the other two terms by first using Hölder's inequality and then Pinkser's inequality
\begin{align}
    \Strategy{1}^{\BRIndex} \Rewards{}(\Strategy{2}^\RegularizationWeight - \Strategy{2}^M) &\leq |\Strategy{1}^{\BRIndex} \Rewards{}(\Strategy{2}^\RegularizationWeight - \Strategy{2}^M)| \nonumber \\
    &\leq ||\Strategy{1}^{\BRIndex}\Rewards{}||_{\infty} |(\Strategy{2}^\RegularizationWeight - \Strategy{2}^M)| \nonumber\\
    & \leq ||\Rewards{}||_{\infty} |(\Strategy{2}^\RegularizationWeight - \Strategy{2}^M)|  \nonumber \\
    &\leq ||\Rewards{}||_{\infty} \sqrt{2 \KLDiv{\Strategy{2}^\RegularizationWeight}{\Strategy{2}^M}}\\
   (\Strategy{1}^{M} -\Strategy{1}^{\RegularizationWeight})\Rewards{}\BRStrategy{2}    &\leq ||\Rewards{}||_{\infty} \sqrt{2 \KLDiv{\Strategy{1}^\RegularizationWeight}{\Strategy{1}^M}}
\end{align}

We can further use inequality $\sqrt{x} + \sqrt{y} \leq \sqrt{2} \sqrt{x + y}$
\begin{align}
    \Exploitability(\Strategy{}^{\RegularizationWeight}) &\leq \Exploitability(\Strategy{}^M) +  ||\Rewards{}||_{\infty} \sqrt{2 \KLDiv{\Strategy{1}^\RegularizationWeight}{\Strategy{1}^M}} +  ||\Rewards{}||_{\infty} \sqrt{2 \KLDiv{\Strategy{1}^\RegularizationWeight}{\Strategy{1}^M}} \\
    & \leq \Exploitability(\Strategy{}^M) + ||\Rewards{}||_{\infty}  \sqrt{ 2 \cdot(2 \KLDiv{\Strategy{1}^\RegularizationWeight}{\Strategy{1}^M} + 2 \KLDiv{\Strategy{2}^\RegularizationWeight}{\Strategy{2}^M})}
\end{align}
Since the used regularizer $\Regularizer$ is entropy, we can use \cref{eq:mmd_expl_bound} to show
\begin{equation}
    \Exploitability(\Strategy{}^{\RegularizationWeight}) \leq \Exploitability(\Strategy{}^M) + ||\Rewards{}||_{\infty} \sqrt{\frac{4 \Exploitability(\Strategy{}^M)}{\RegularizationWeight}}
\end{equation}
\end{proof}
 
\section{Experimental details}
\label[appendix]{app:exp_details}
The codebase for this paper is written in Python 3.13, with libraries from the Jax ecosystem \citep{jax2018github,flax2020github,deepmind2020jax}. Every experiment run used a single core of an AMD EPYC 7543 with 16 GB of RAM and a single Nvidia A100.

\subsection{Game rules}

\paragraph{Biased Matching Pennies} Biased Matching Pennies is a one-shot game, in which players choose a side of a coin, either Heads or Tails. If the choices match, Player 1 wins; otherwise, Player 2 wins. In this version, the rewards are doubled if both choose Tails. The normal-form representation of this game is in \cref{tab:biased_mp}. In our experiments, we treat this as a sequential game in which Player 2 acts first, followed by Player 1, who did not observe Player 2's choice.

\paragraph{Rock-Paper-Scissors} Rock-Paper-Scissors is a one-shot game where each player selects one of three actions: Rock, Paper, or Scissors. An identical choice results in a draw, rock defeats scissors, paper defeats rock, and scissors defeats paper. The normal-form representation is in \cref{tab:rps}. As with Matching Pennies, we implement this sequentially, with Player 2 acting first and Player 1 acting with no knowledge of the opponent's selection.

\begin{table}[h]
    \centering 
    \caption{Normal-form representation of the games used in \cref{sec:tiny_games}}
    \label{tab:matrix_games}
    \begin{subtable}[b]{0.45\textwidth}
        \caption{Biased Matching Pennies}
        \label{tab:biased_mp}
        \centering
        \[
        \begin{array}{cc|c|c|}
         & \multicolumn{1}{c}{} & \multicolumn{2}{c}{\text{P2}} \\
         & \multicolumn{1}{c}{} & \multicolumn{1}{c}{\text{H}} & \multicolumn{1}{c}{\text{T}} \\
        \cline{3-4}
        & \text{H} & 1 & 0 \\
        \cline{3-4}
         \text{P1}  & \text{T} & 0 & 2\\
        \cline{3-4} 
        \end{array}
        \]
    \end{subtable}
    \hfill 
    \begin{subtable}[b]{0.45\textwidth}
        \caption{Rock-Paper-Scissors}
        \label{tab:rps}
        \centering
        \[
        \begin{array}{cc|c|c|c|}
         & \multicolumn{1}{c}{} & \multicolumn{3}{c}{\text{P2}} \\
         & \multicolumn{1}{c}{} & \multicolumn{1}{c}{\text{R}} & \multicolumn{1}{c}{\text{P}} & \multicolumn{1}{c}{\text{S}} \\
        \cline{3-5}
         & \text{R} & 0 & -1 & 1 \\
        \cline{3-5}
         \text{P1}  & \text{P} & 1 & 0 & -1 \\
        \cline{3-5}
         & \text{S} & -1 & 1 & 0 \\
        \cline{3-5}
        \end{array}
        \]
    \end{subtable}
\end{table}

\paragraph{Battleship} Battleship $S, L$ is a two-player game which is played on two private $S\times S$ grids. Each player has a set of ships $L$ of varying lengths (width is always 1). In the first phase of the game, players secretly place their ships on their respective grids, so that no ship can be placed in the 8-neighborhood of any other ship. Then, in the second phase, players take turns firing shots at the opponent's board, receiving only feedback on whether the ship was hit or not. When all segments of a ship are hit, the ship is publicly declared sunken. The game concludes when one player sinks all of the opponent's ships. An example of notation used in the paper is Battleship $7, [4, 3, 3, 2]$, which is played on $7 \times 7$ grids with ships of sizes 4, 3, 3, 2.

\paragraph{Goofspiel} Imperfect Information Goofspiel $X$ is a card game, where each player has cards valued from 1 to $X$, each only once. In each of the $X$ rounds, a prize card from the dealer's deck is revealed. Players submit a secret bid using one card from their hand, and the player with the highest bid wins points equal to the value of the prize card. In the event of a tie, no points are awarded. Bids are not publicly revealed, only the outcome of the round. In our version, the dealer's deck is not random, but it contains prize cards in descending order from $X$ to 1.

\paragraph{Leduc Hold 'em} Leduc Hold 'em is a simplified poker variant played with 6 cards (2 Jacks, 2 Queens, 2 Kings). At the beginning of the game, each player is dealt a single card. The game then progresses through 2 betting rounds. After the first round, the public card is revealed. Each betting round can consist of at most 2 raises. If any player has the same card as the public one, they win. The player with the higher-ranked card wins, where the ranking order is: Jack $<$ Queen $<$ King.


\subsection{Subgame solving}
To show that additional training using the gadget game is a minor change to the training algorithm, we provide pseudocode for this approach in \cref{alg:gadget}, highlighting the modifications in red. This pseudocode also shows how the Bayesian unsafe approach is used, as it differs from the blueprint training only in sampling world states using the belief model rather than starting from the initial state.
\begin{algorithm}
\caption{Subgame Solving}
\label{alg:gadget}
\begin{algorithmic}[1]
\Require Blueprint parameters $\NeuralParameters$, resolving player $\Player$, opponent $\OtherPlayer$, information set $\Infoset{\Player}$
\State All techniques, \textcolor{red}{Gadget}
\State $\GadgetParameters \gets $ Initialize gadget actor parameters 
\State $\NeuralParameters' \gets \NeuralParameters$
\State $\PublicState \gets \PublicState(\Infoset{\Player})$ 
\While{Available budget}
    \State Sample a $\WorldState$ world state using the belief model \hfill \textit{Initial state in the training}
    \State  Sample trajectory $\Trajectory$ from $\WorldState$ using $\Strategy{\NeuralParameters'}$
    \State Estimate counterfactual values $\CounterfactualValue{k, \text{target}}{\Strategy{\NeuralParameters'}}, q_{k, \text{target}}^{\Strategy{\NeuralParameters'}}$ for each player $k$ in every encountered infoset in the trajectory $\Trajectory$
    \State \textcolor{red}{$q_{\mathrm{term}} \gets \CounterfactualValue{\OtherPlayer, \NeuralParameters}{}\!\left(\Infoset{\OtherPlayer}(\WorldState)\right)$ \hfill \textit{Blueprint critic value}}
    \State \textcolor{red}{$q_{\mathrm{cont}} \gets \CounterfactualValue{\OtherPlayer, \text{target}}{\Strategy{\NeuralParameters'}}(\Infoset{\OtherPlayer}(\WorldState))$ \hfill \textit{Target value}}
    \State \textcolor{red}{Update $\GadgetParameters$ with policy-gradient loss using $\left(q_{\mathrm{term}},\, q_{\mathrm{cont}}\right)$ at $\Infoset{\OtherPlayer}(\WorldState)$}
    \State \textcolor{red}{$p_{\mathrm{cont}} \gets \Strategy{\OtherPlayer, \GadgetParameters}\!\left(\Infoset{\OtherPlayer}(\WorldState),\, \mathrm{cont}\right)$ }
    \State Update $\NeuralParameters'$ with actor and critic losses \textcolor{red}{scaled by $p_{\mathrm{cont}}$}
\EndWhile
\State \Return $\NeuralParameters'$
\end{algorithmic}
\end{algorithm}

\subsection{Belief model}
We have conducted an additional experiment to verify the quality of the samples from the trained belief models used in \cref{sec:small_games}. We have extracted the reaches from the trained blueprint for each possible state, and we have also extracted the probability distributions from the transformer's logits. We then compared those two distributions using Jensen-Shannon divergence, and we show the mean and maximum of those divergences from each public state in \cref{fig:js_divergence}.

Surprisingly, the belief model trained for the gadget game, which uses importance sampling correction, sometimes produces more accurate samples even with the added variance. This mostly happens with stronger blueprints in Goofspiel 5, where public states later in the game contain hundreds of states. The likely cause is that, in Bayesian solving, the reaches are closer to zero than in the gadget game, which makes them harder to predict correctly due to the variance introduced by the sampling itself.

\begin{figure}
    \centering
    \includegraphics[width=0.99\linewidth]{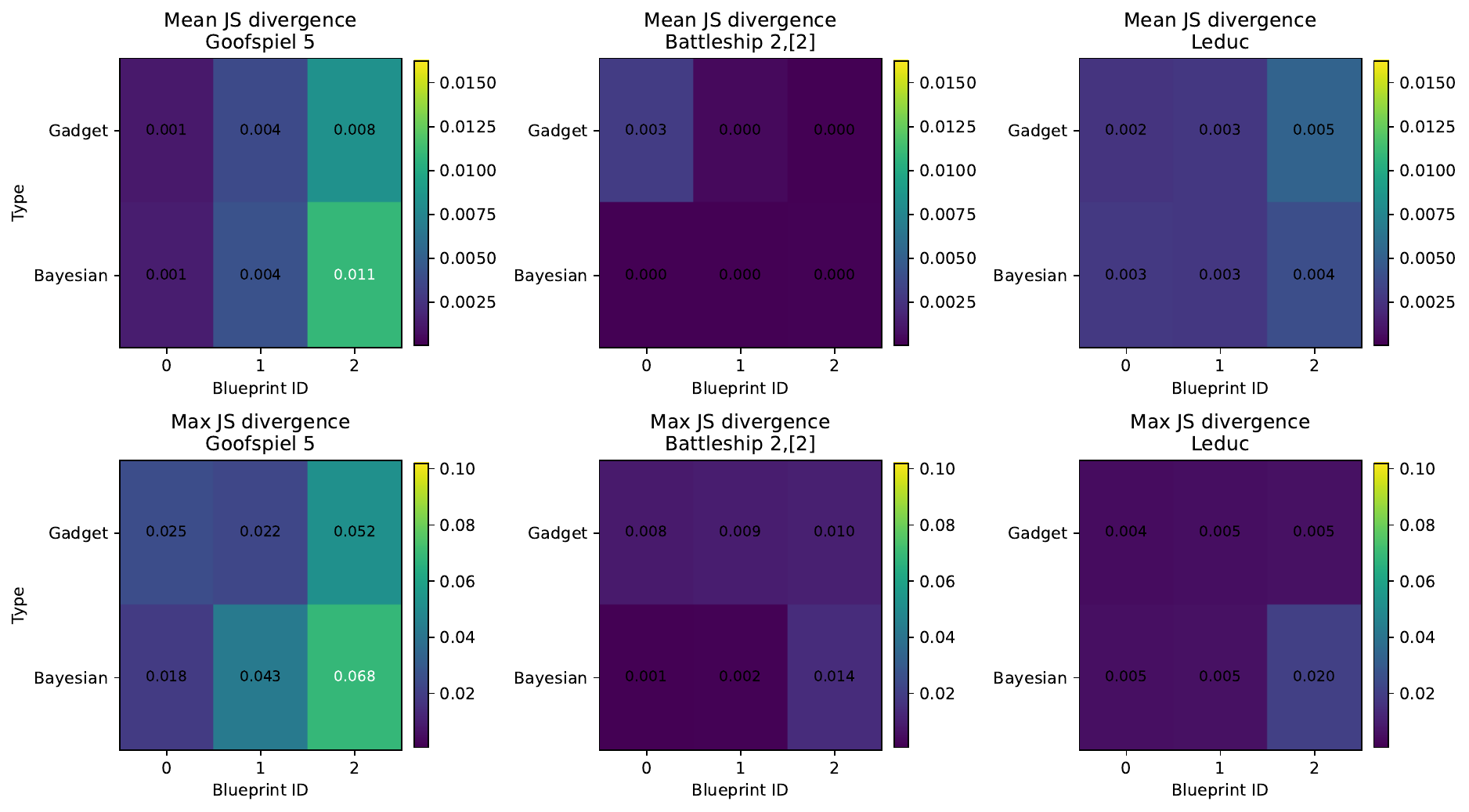}
    \caption{Mean and maximum of Jensen-Shannon divergence between the actual blueprint reaches and the samples from the belief model across each public state in the game.}
    \label{fig:js_divergence}
\end{figure}

\subsection{Choice of the gadget}
In this paper, we have used the resolving gadget game, but there is another type of gadget game \textit{max-margin gadget game}, which further improves upon the resolving gadget game by maximizing the minimal improvement in the subgame \citep{moravcik2016maxmargin}. Recent work has shown that the resolving gadget game may contain stronger equilibria than the max-margin \citep{kubicek2026refinements}. 

In this work, we did not use the max-margin because it switches the order of gadget decision. In a resolving gadget game, the chance first decides the state and only then does the opponent select its belief. In max-margin, the opponent selects its belief, and only then is the particular state selected. This means that the opponent needs to decide belief across all of its information sets at the root of the subgame $\Subgame{\Infoset{}}$. However, in large games, it is not known a priori how many information sets there are at the root of the subgame. Simulating the max-margin gadget would require this knowledge or having a way to add new actions to the network on the fly. On the other hand, the resolving gadget only adds a decision with 2 actions, so the neural network can easily simulate it.

Consider a situation in which the opponent reaches a part of the subgame only through suboptimal play. It gives you the ``gift'' in utility you can get in that subgame. \textit{Reach max-margin} subgame solving redistributes this gift to allow more gains in the parts of the subgame that can be reached in optimal play~\citep{brown2017reachmaxmargin}. The gifts can be combined with any gadget game. However, in this work, we did not explore this combination with our simulation and leave this for future research.

\subsection{Hyperparameters}
We have used two sets of hyperparameters, one set for small experiments in \cref{sec:tiny_games,sec:small_games} and the second set for large experiments \cref{sec:large_games}. Both are in  \cref{tab:hyperparameters}

\begin{table}[htbp]
    \centering
    \caption{Hyperparameters in all experiments}
    \label{tab:hyperparameters}
    \begin{tabular}{llll}
        \toprule  \textbf{Parameter} & \textbf{Small} & \textbf{Large} \\
        \midrule 
            RNaD regularization $\eta$  & $0.2$ & $0.2$\\
            RNaD regularization policy change & $2000$ & $20000$\\
            NeuRD $\beta$ clip  & $2$ & $2$ \\
            NeuRD avantage clip  & $5$ & $5$\\
            MLP hidden sizes & $256$ & $1024$\\
            MLP hidden layers & $2$ & $2$ \\
            Belief transformer heads & $2$ & $4$\\
            Belief transformer layers & $4$ & $8$ \\
            Belief transformer dense size & $128$ & $256$\\
            Activation function & GELU & GELU \\
            Normalization function & RMSNorm & RMSNorm \\
            Batch size & 64 & 256 \\
            Strategy learning rate & $3 \cdot 10^{-4}$ & $1 \cdot 10^{-4}$ \\
            Belief learning rate & $1 \cdot 10^{-3}$ & $1 \cdot 10^{-3}$ \\
            Gadget learning rate & $1 \cdot 10^{-3}$ & $1 \cdot 10^{-3}$ \\
            Optimizer & ADAM & ADAM \\
            ADAM $\beta_1, \beta_2$ & $0.99, 0.999$ & $0.99, 0.999$ \\
            V-trace $\lambda$ & $0.95$  & $0.95$\\
            Discount Factor $\gamma$ & $1.0$  & $1.0$\\
        \bottomrule
    \end{tabular}
\end{table}

\subsection{Computational resources}
\paragraph{One-shot games} For each game in this category, we trained five independent blueprint strategies and their corresponding belief models. Subsequently, test-time reasoning was run for each blueprint. The total computational time for these experiments was less than 24 GPU-hours per game.
\paragraph{Small games} For each game, we utilized a single training run and selected three checkpoints representing different stages of strategy convergence. For each checkpoint, we trained three distinct belief models: the first is used for Bayesian solving, and samples states conditioned on the public state based on the probability of reaching those states when following the blueprint. The second is used for gadget games and also samples states conditioned on the public state, but it is unbiased with respect to the opponent's strategy. The third is used for KLSS and UEF, and it samples states conditioned on the information state. The initial training phase required approximately 12 GPU-hours per game. For the evaluation phase, we allocated 2, 5, or 10 seconds of training time per decision point across five different seeds. A single evaluation run was completed in under 6 hours. This makes the total time for this experiment $3\cdot(12 + 5\cdot 6 \cdot 3) = 306$ GPU-hours.
\paragraph{Large games} We trained five independent blueprint strategies for each large-scale game, each requiring 72 GPU-hours. This was followed by an additional 72 GPU-hours per blueprint to train the three belief model variants. Finally, the head-to-head evaluations lasted approximately 100 hours per configuration. The total computational time is estimated to be at most $3\cdot(5\cdot (72 + 72 \cdot 3) + 2\cdot4\cdot100) = 6720$ GPU-hours. The games played in that time-frame for each bar in \cref{fig:large_games} are in \cref{tab:large_games}

\begin{table}[!h]
\begin{tabular}{l|c|c}
Game                           & Time & Playthroughs \\  \hline
\multirow{2}{*}{Battleship 7}  & 10 s & \textasciitilde $1500$            \\
                               & 30 s & \textasciitilde$ 500$          \\ \hline
\multirow{2}{*}{Battleship 10} & 10 s & \textasciitilde$ 800$           \\
                               & 30 s & \textasciitilde$ 300$           \\ \hline
\multirow{2}{*}{Goofspiel 10}  & 10 s & \textasciitilde$ 1500$          \\
                               & 30 s & \textasciitilde$ 1300$           
\end{tabular}
\centering
\caption{Amount of playthroughs used for each bar in \cref{fig:large_games}.}
\label{tab:large_games}
\end{table}

\section{Additional Experiments}
We provide several additional experiments to accompany the main results from the paper
\subsection{Effect of magnet strength in unsafe resolving}
To verify the effects of the bounded strategy change from \cref{thm:fixed_magnet_kl_divergence}, we conducted the same experiment as in \cref{sec:small_games} in Rock-paper-scissors, by using the varied strengths of magnets to solve the same subgame. We show the results of this experiment in \cref{fig:varying_magnet}. 

The blueprints used have a small bias towards one action, making the optimal Bayesian subgame solution a pure strategy that is fully exploitable. The magnet then prevents the strategy from changing too much, and, as expected, the lower the magnet regularization, the closer the algorithm converges to the best response.

\begin{figure}[!h]
    \centering
    \includegraphics[width=0.5\linewidth]{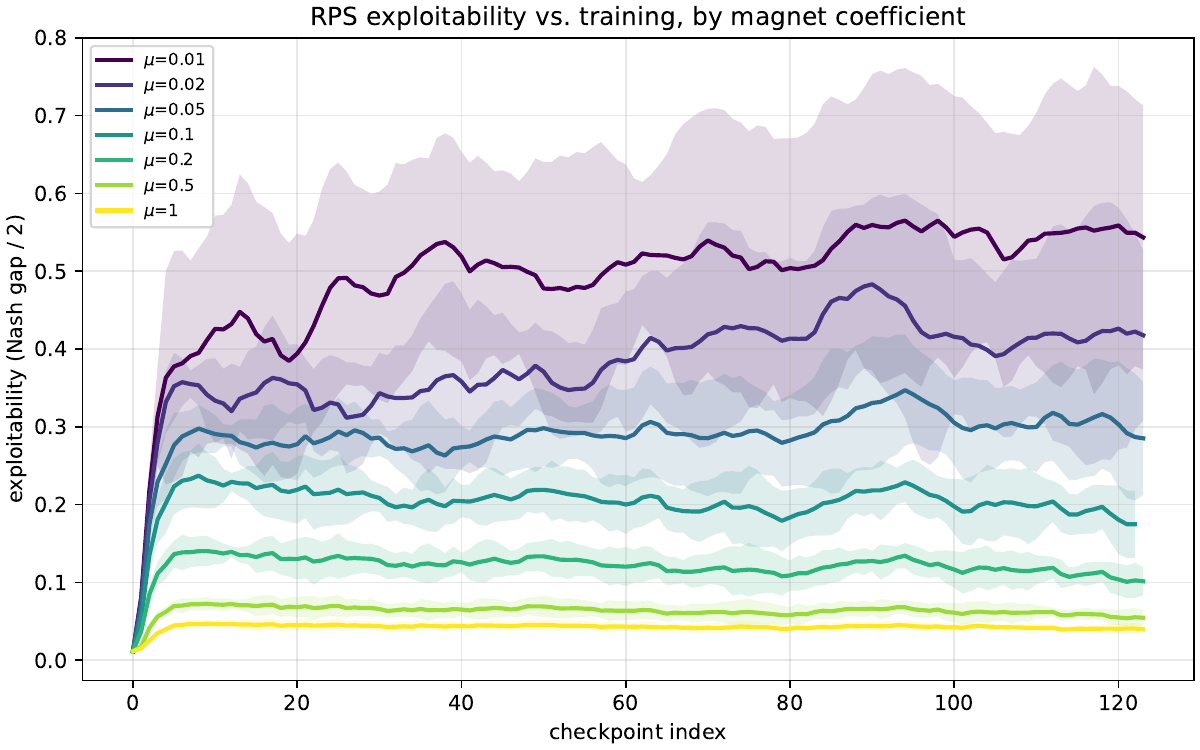}
    \caption{Exploitability in Rock-paper-scissors, when solving a single subgame using Bayesian subgames with fixed magnet, for varying strengths of the magnet.}
    \label{fig:varying_magnet}
\end{figure}

\subsection{Head-to-head performance against other opponents}
In addition to the experiments in \cref{sec:large_games}, we also matched all the subgame-solving techniques and blueprints used against two types of players in Battleship 7,[4,3,3,2]. First, being a uniform random player, which picks an action at random, and second, being a strong heuristic player that uses strategies often employed by strong human players. 

The heuristic player randomly positions its ships. Then, during the shooting phase, it always shoots a checkboard pattern based on the smallest opponent's ship. When the ship is found, the player tries to sink it as quickly as possible, then continues with the checkboard pattern shooting.

We show the results with 95\% confidence intervals in \cref{fig:large_games_add}. The results are in line with those from \cref{sec:large_games}, where Gadget and both Bayesian approaches significantly improve upon the blueprint, but the performance improvement among them is inconclusive. Although UEF also improves over the blueprint, the gains are smaller than those of the other methods. 

\begin{figure*}
    \centering    
    \begin{subfigure}[b]{0.49\textwidth}
        \centering
    \includegraphics[width=0.99\linewidth]{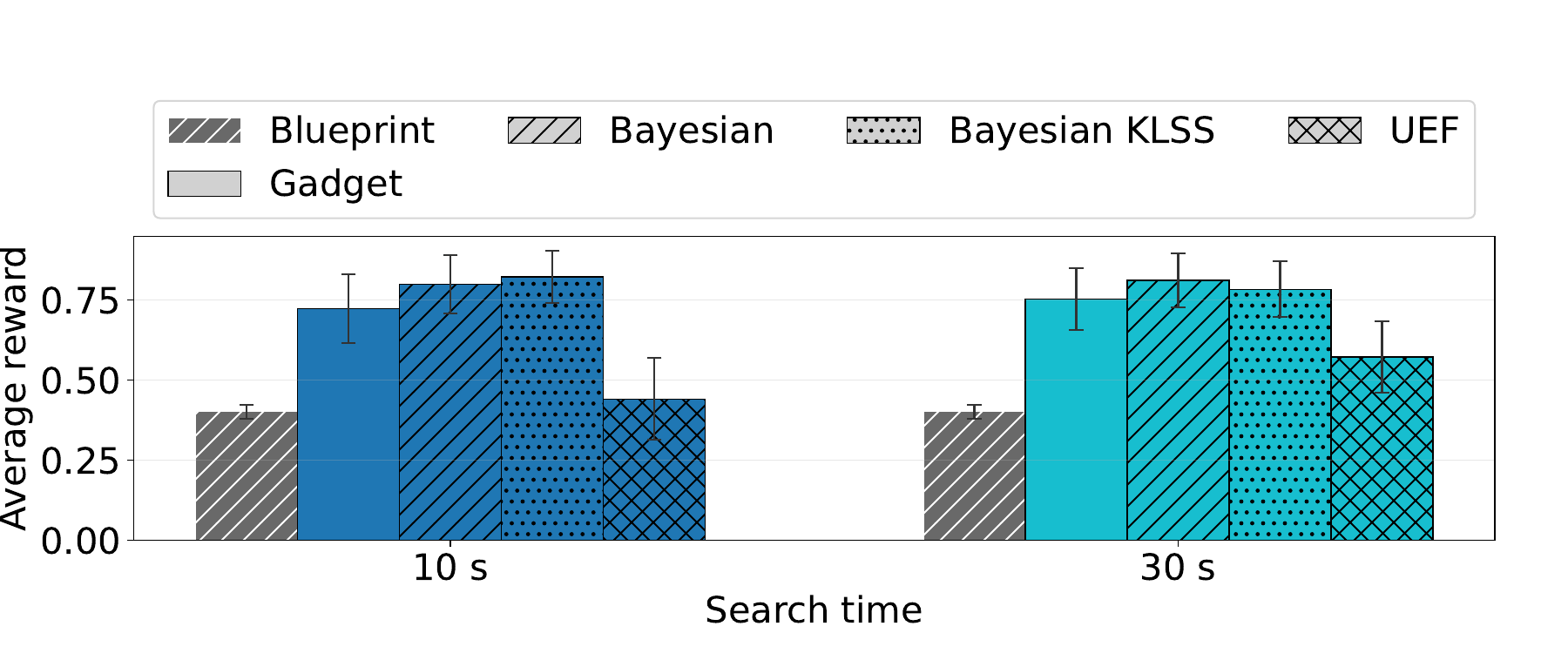}
        \caption{Random opponent}
        \label{fig:bs7_random}
    \end{subfigure}
    \hfill 
    \begin{subfigure}[b]{0.49\textwidth}
        \centering
    \includegraphics[width=0.99\linewidth]{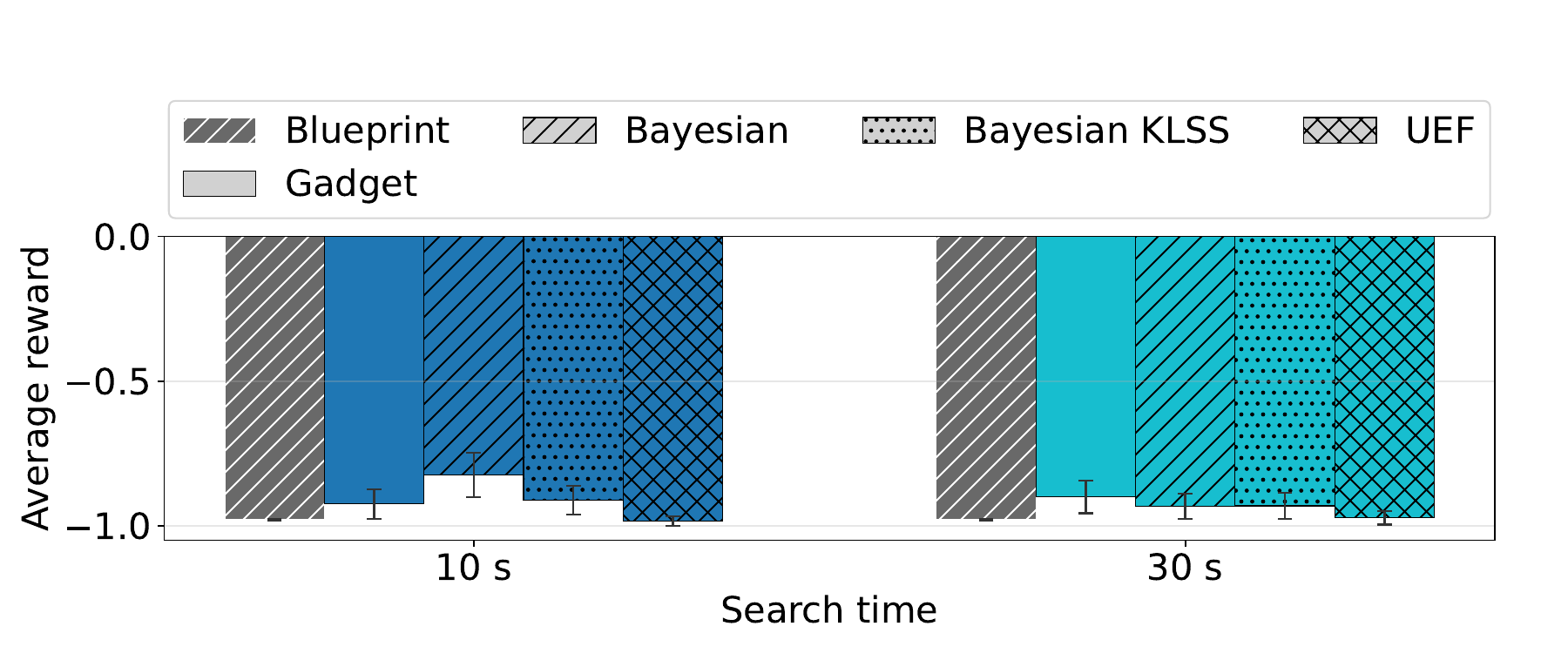}
        \caption{Heuristic opponent}
        \label{fig:bs7_heuristic}
    \end{subfigure}  
    \caption{Head-to-head average reward with 95\% confidence intervals of blueprint and different subgame solving techniques against the different types of opponents in Battleship 7,[4,3,3,2].}
    \label{fig:large_games_add} 
\end{figure*} 
\section{LLM usage}
We have used Claude Sonnet 4.6 (with Claude Code) and Cursor Composer 2 during the implementational part of our work. We have used Claude Sonnet 4.6 and Gemini 3 during the writing of this paper. These models were not used to introduce novel ideas into the paper, but to refine the writing. Besides that, the models were used as a help in constructing the proofs in the paper. All outputs from those models have been double-checked before inclusion in the final version.




\end{document}